\documentclass[12pt]{article}
\usepackage{amsmath}
\usepackage{amssymb}
\usepackage{amsthm} 
\usepackage{fullpage}
\usepackage[utf8]{inputenc}
\usepackage{mathtools}
\usepackage{url}
\usepackage[usenames,svgnames]{xcolor}
\usepackage{cite}
\usepackage{enumitem}
\usepackage[titletoc,title]{appendix}

\usepackage[pagebackref=false]{hyperref} 
\usepackage[all]{hypcap} 

\theoremstyle{plain}
\newtheorem{theorem}{Theorem}[section]
\newtheorem{lemma}[theorem]{Lemma}

\newtheorem{proposition}[theorem]{Proposition}
\newtheorem{corollary}[theorem]{Corollary}
\newtheorem{example}{Example}[section]
\newtheorem{examples}{Example}[subsection]

\newtheorem{remark}{Remark}[section]

\theoremstyle{definition}
\newtheorem{definition}{Definition}[section]

\numberwithin{equation}{section} 

\hypersetup{
breaklinks,
colorlinks,
citecolor=Green,
linkcolor=BlueViolet,
urlcolor=DarkBlue,
pdftitle={Weighted Hurwitz numbers , constellations and topological recursion},
pdfauthor={A. Alexandrov, G. Chapuis, B. Eynard, J. Harnad } }

\DeclareMathOperator{\End}{End}
\DeclareMathOperator{\Span}{span}

\DeclareMathOperator{\cyc}{cyc}

\DeclareMathOperator{\GL}{GL}

\DeclareMathOperator{\Li}{Li}
\DeclareMathOperator{\aut}{aut}

\def\ra{{\rightarrow}}

\def\det{\mathrm {det}}

\def\Spin{\mathrm {Spin}}
\def\End{\mathrm {End}}
\def\span{\mathrm {span}}

\def\ln{\mathrm {ln}}

\def\span{\mathrm {span}}

\def\Gr{\mathrm {Gr}}
\def\res{\mathop{\mathrm {res}}\limits}

\DeclarePairedDelimiter{\abs}{|}{|}

\DeclarePairedDelimiter{\no}{:}{:}

\def\be{\begin{equation}}
\def\ee{\end{equation}}

\def\bea{\begin{eqnarray}}
\def\eea{\end{eqnarray}}

\def\bt{\begin{theorem}}
\def\et{\end{theorem}}

\def\bex{\begin{example}\small \rm}
\def\eex{\end{example}}

\def\bexs{\begin{examples}\small \rm}
\def\eexs{\end{examples}}

\def\ra{\rightarrow}
\def\ss{\subset}
\def\deq{\coloneqq}

\def\br{\begin{remark}\small \rm}
\def\er{\end{remark}}

\def\res{\mathop{\mathrm{res}}\limits}

\def\&{&{\hskip -20pt}}

\def\BB{\mathcal{B}}
\def\CC{\mathcal{C}}
\def\DD{\mathcal{D}}
\def\EE{\mathcal{E}}
\def\FF{\mathcal{F}}
\def\JJ{\mathcal{J}}
\def\HH{\mathcal{H}}
\def\II{\mathcal{I}}

\def\MM{\mathcal{M}}

\def\OO{\mathcal{O}}

\def\Cb{\mathbf{C}}

\def\Ib{\mathbf{I}}

\def\Nb{\mathbf{N}}

\def\Nb{\mathbf{N}}
\def\Pb{\mathbf{P}}
\def\Qb{\mathbf{Q}}

\def\Zb{\mathbf{Z}}

 \def\grg{\mathfrak{g}}

 \def\grl{\mathfrak{l}}

\def\grP{\mathfrak{P}}

\def\psistar{\psi^{*}}

\def\normord{:}
\def\pp{\partial}
\def\normordboson{ {\scriptstyle {{*}\atop{*}}} }

\begin{document}
\baselineskip 16pt
\medskip
\begin{center}
\begin{Large}\fontfamily{cmss}
\fontsize{17pt}{27pt}
\selectfont
\textbf{Fermionic approach to weighted Hurwitz numbers and topological recursion}
\end{Large}
\\
\bigskip \bigskip
\begin{large}  A. Alexandrov$^{1, 2, 4}$\footnote{e-mail: alexandrovsash@gmail.com},   G. Chapuy$^{5}$\footnote{e-mail: guillaume.chapuy@liafa.univ-paris-diderot.fr}, 
B. Eynard$^{2, 6}$\footnote{e-mail:  bertrand.eynard@cea.fr}  and J. Harnad$^{2, 3}$\footnote {e-mail: harnad@crm.umontreal.ca  }
 \end{large}\\
\bigskip
\begin{small}
$^{1}${\em Center for Geometry and Physics, Institute for Basic Science (IBS), Pohang 37673, Korea}\\
 \smallskip
$^{2}${\em Centre de recherches math\'ematiques,
Universit\'e de Montr\'eal\\ C.~P.~6128, succ. centre ville, Montr\'eal,
QC H3C 3J7  Canada}\\
 \smallskip
$^{3}${\em Department of Mathematics and Statistics, Concordia University\\ 1455 de Maisonneuve Blvd.~W.~Montreal, QC H3G 1M8  Canada}\\
\smallskip
$^{4}${\em ITEP, Bolshaya Cheremushkinskaya 25, 117218 Moscow, Russia} \\
\smallskip
$^{5}${\em CNRS UMR 7089, Universit\'e Paris Diderot \\
 Paris 7, Case 7014
75205 Paris Cedex 13 France}\\
 \smallskip
 $^{6}${\em Institut de Physique Th\'eorique, CEA, IPhT, \\ F-91191 Gif-sur-Yvette, 
 France \\
 CNRS   URA 2306, F-91191 Gif-sur-Yvette, France} 
\smallskip
\end{small}
\end{center}

\begin{small}\begin{center}
\today
\end{center}\end{small}

\begin{abstract}
   A fermionic representation is given for all the quantities entering in the generating function  approach to weighted Hurwitz  numbers and topological recursion. This includes: KP and $2D$ Toda $\tau$-functions of hypergeometric  type, which  serve as generating functions  for weighted single and double Hurwitz numbers; the Baker function, which is expanded in an  adapted basis obtained by  applying the same dressing transformation to all vacuum basis elements; the multipair correlators and the multicurrent correlators.  Multiplicative recursion relations and a linear differential system are deduced for the adapted bases and their duals, and a Christoffel-Darboux type formula is derived for the pair correlator. The quantum and classical spectral curves linking this theory with the topological recursion program are derived, as well as the  generalized {\it cut and join} equations. The results are detailed for four special cases: the  simple single and double Hurwitz numbers, the weakly monotone case, corresponding to signed enumeration of coverings,  the strongly monotone case, corresponding to Belyi curves and the simplest version of quantum weighted Hurwitz numbers.  
     \end{abstract}

\break

\tableofcontents

\section{Weighted Hurwitz numbers and generating $\tau$-functions }

It has been known since the work of Pandharipande and Okounkov \cite{Pa, Ok}  that certain special classes of $\tau$-functions of the KP and $2D$-Toda hierarchies can serve as generating functions for various types of weighted Hurwitz numbers. It also is known that the computation
of these enumerative geometrical/combinatorial invariants may be embedded into the framework of Topological Recursion  \cite{ BEMS, EO1, EO2, MSS}. Fermionic representations have played a key r\^ole in simplifying the calculations and clarifying the origin of these constructions. The present work is part of an
ongoing series (see \cite{ACEH1, ACEH2}) devoted to the inclusion of the generating function approach to weighted Hurwitz numbers of the most general type  into the framework of Topological Recursion.

\subsection{Weighted Hurwitz numbers}
\label{tau_weighted_hurwitz}

For a set of partitions $\{\mu^{(i)} \}_{i=1,\dots, k}$ of weight $|\mu^{(i)}|=N$, 
the Hurwitz numbers   $H(\mu^{(1)}, \dots, \mu^{(k)})$ are  defined geometrically  \cite{Hu1, Hu2}  as the number
of inequivalent $N$-fold branched coverings  $\CC \ra \Pb^1$ of the Riemann sphere having $k$ branch points with
ramification profiles given by the partitions $\{\mu^{(1)}, \dots, \mu^{(k)}\}$, 
normalized by the inverse  $1/|\aut (\CC)|$ of the order of the automorphism group of the covering. 
 An equivalent combinatorial/group theoretical  definition \cite{Frob1, Frob2, Sch} is that  $H(\mu^{(1)}, \dots, \mu^{(k)})$ is 
 ${1/ N !}$ times the  number of distinct   factorizations of the identity  element $\Ib \in S_N $ in the symmetric 
 group in $N$ elements as a product of $k$ factors $h_i$, belonging to the respective conjugacy classes $\cyc(\mu^{(i)})$
\be
\Ib = h_1, \cdots h_k, \quad h_i \in \cyc(\mu^{(i))}).
\label{hurwitz_factoriz}
\ee
The equivalence of the two follows from  the monodromy homomorphism from the fundamental group of $\Pb^1$ 
punctured at the branch points into $S_N$ obtained by lifting closed loops from the base to the covering.

Denoting by
\be
\ell^*(\mu) := |\mu| - \ell(\mu)
\label{colength}
\ee
the {\em colength} of the partition $\mu$ (the difference between the weight and the length),
the Riemann Hurwitz theorem relates the Euler characteristic $\chi$ of the covering curve $\CC$ to the sum
of the colengths $\ell^*(\mu^{(i)})$ of the ramification profiles at the branch points
\be
\chi  = 2 - 2g = 2N - \sum_{i=1}^k \ell^*(\mu^{(i)}),
\label{riemann_hurwitz}
\ee
where $g$ is the genus of $\CC$.

 A weight generating function $G(z)$ and its dual $\tilde{G}(z) := 1/G(-z)$, may be represented 
 in terms of an infinite set of constants (or indeterminates) 
  \be
 {\bf c}:= (c_1, c_2, \dots ),
 \label{constants}
 \ee
either as an infinite product, or an infinite sum:
 \bea
G(z) &\& := \prod_{i=1}^\infty (1+ z c_i) = 1+ \sum_{i=1}^\infty g_i z^i
\label{gener_G}
 \\
\tilde{G}(z) &\& := \prod_{i=1}^\infty (1- z c_i)^{-1} = 1+ \sum_{i=1}^\infty \tilde{g}_i z^i, 
\label{gener_tildeG}
 \eea
where the coefficients $\{g_i\}_{i=1, \dots, \infty}$, $\{\tilde{g}_i\}_{i=1, \dots, \infty}$ in the (formal) Taylor series are,
respectively, the elementary and complete symmetric functions
\be
g_i = e_i({\bf c}), \quad \tilde{g}_i = h_i({\bf c}).
\label{gi_tildegi}
\ee
We define the (associative, commutative) algebra 
\be
\mathbb K = \mathbb Q[g_1,g_2,g_3,\dots]
\ee
of polynomials in the $g_i$'s over the field $\mathbb Q$ of rationals, which is the same (when the number of $c_i$'s is finite) as the algebra $\Lambda$ of symmetric functions of the $c_i$'s (over the coefficient field $\mathbb Q$).

Given a pair of partitions $(\mu, \nu)$ of $N$, the weighted double and single Hurwitz numbers are defined by the sums  \cite{GH2, H2} 
   \bea
H^d_G(\mu, \nu) &\&\deq \sum_{k=0}^d \sideset{}{'}\sum_{\substack{\mu^{(1)}, \dots, \mu^{(k)} \\ \sum_{i=1}^k \ell^*(\mu^{(i)})= d}}
W_G(\mu^{(1)}, \dots, \mu^{(k)}) H(\mu^{(1)}, \dots, \mu^{(k)}, \mu, \nu) 
\label{Hd_G}
\\
H^d_{\tilde{G}}(\mu, \nu) &\&\deq \sum_{k=0}^d\sideset{}{'}\sum_{\substack{\mu^{(1)}, \dots, \mu^{(k)} \\ \sum_{i=1}^k \ell^*(\mu^{(i)})= d}}
W_{\tilde{G}}(\mu^{(1)}, \dots, \mu^{(k)})  H(\mu^{(1)}, \dots, \mu^{(k)}, \mu, \nu), \\
H^d_G(\mu)&\& := H^d_G(\mu, \nu= (1)^N) , \qquad H^d_{\tilde{G}}(\mu) := H^d_{\tilde{G}}(\mu, \nu= (1)^N) 
\label{Hd_tildeG}
\eea
where $\sideset{}{'}\sum$ denotes the sum over all  $k$-tuples of partitions 
$\{\mu^{(1)}, \dots, \mu^{(k)}\}$ other than the cycle type of the identity element $(1^N)$, the weights 
\be
W_G(\mu^{(1)}, \dots, \mu^{(k)}):= m_\lambda ({\bf c}) =
\frac{1}{\abs{\aut(\lambda)}} \sum_{\sigma\in S_k} \sum_{1 \le i_1 < \cdots < i_k}
 c_{i_{\sigma(1)}}^{\lambda_1} \cdots c_{i_{\sigma(k)}}^{\lambda_k},
 \label{m_lambda}
\ee
and 
\be
W_{\tilde{G}}(\mu^{(1)}, \dots, \mu^{(k)}):= f_\lambda ({\bf c}) =
\frac{(-1)^{\ell^*(\lambda)}}{\abs{\aut(\lambda)}} \sum_{\sigma\in S_k} \sum_{1 \le i_1 \le \cdots \le i_k} 
c_{i_{\sigma(1)}}^{\lambda_1},  \cdots c_{i_{\sigma(k)}}^{\lambda_k},
 \label{f_lambda}
\ee 
are, respectively,  the monomial sum and ``forgotten'' symmetric function bases \cite{Mac} of the algebra 
$\Lambda$ of  symmetric functions in the variables $(c_1, c_2, \dots )$  and 
the partition $\lambda$ labelling the coefficients $m_\lambda ({\bf c})$ and $f_\lambda ({\bf c})$ in the sums 
have weight $|\lambda|=d$ and length  $\ell(\lambda) =k$,  with parts $\{\lambda_i\}$  equal to the 
colengths $\{\ell^*(\mu^{(i)})\}$, expressed in weakly decreasing order. The normalization factor
$ \abs{\aut(\lambda)}$ is the order of the automorphism group in $S_d$ of the elements in the 
conjugacy class of cycle type $\lambda$, given by
\be
\abs{\aut(\lambda)} = \prod_{i=1}^d m_i(\lambda)!,
\ee
where $m_i(\lambda)$ is the number of parts of $\lambda$ (i.e. the number of colengths $\ell^*(\mu^{(j)})$) 
equal to $i$.

\begin{remark}
By convention,  if the partitions are not all of the same weight,  we set
\be
H(\mu^{(1)}, \dots, \mu^{(k)}) : = 0
\ee
and thus 
\be
H^d_G(\mu, \nu)= 0 \quad {\rm if}\quad |\mu|\neq |\nu|,
 \quad H^d_{\tilde{G}}(\mu, \nu) = 0 \quad {\rm if}\quad |\mu|\neq |\nu|.
\ee
\end{remark}

An equivalent  definition of $H^d_G (\mu, \nu)$ and $H^d_{\tilde{G}} (\mu, \nu)$   (see \cite{GH2} for the proof) is:
 \bea
H^d_G(\mu, \nu)&\& \deq {1\over N!} \sum_{\lambda, \ \abs{\lambda}=d} e_\lambda({\bf c}) m^\lambda_{\mu \nu}
\label{Fd_G_def}
\\
H^d_{\tilde{G}}(\mu, \nu)&\& \deq  {1\over N!} \sum_{\lambda, \ \abs{\lambda}=d} h_\lambda({\bf c}) m^\lambda_{\mu \nu}
\label{Fd_tilde_G_def}.
\eea
where $\{e_\lambda({\bf c})\}$, $\{h_\lambda({\bf c})\}$ are the elementary and complete
symmetric function bases for $\Lambda$ and $m^\lambda_{\mu \nu}$ is the number of monotonic 
$d=\abs{\lambda}$ step paths in the Cayley graph of $S_N$ 
generated by the transpositions, starting from  an element $h$ in the conjugacy class $\cyc(\nu)$ and ending in $\cyc(\mu)$, 
with {\em signature} $\lambda$   \cite{GH2}; i.e., such that the parts of $\lambda$
are equal to the number of times any given second element $b_i$ in the steps $(a_i, b_i), \ a_i< b_i$ of the path is repeated.
These belong to the algebra $\mathbb K$:
\be
H^d_G(\mu,\nu), \ H^d_{\tilde{G}}(\mu, \nu) \in \mathbb K.
\label{fieldK}
\ee

Letting $\tilde{H}(\mu^{(1)}, \dots, \mu^{(k)})$ denote  the number of connected coverings with  given ramification profiles,
 the corresponding weighted Hurwitz numbers for connected coverings are denoted
   \bea
\tilde{H}^d_G(\mu, \nu) &\&\deq \sum_{k=0}^d \sideset{}{'}\sum_{\substack{\mu^{(1)}, \dots, \mu^{(k)} \\ \sum_{i=1}^k \ell^*(\mu^{(i)})= d}}
m_\lambda ({\bf c})\tilde{H}(\mu^{(1)}, \dots, \mu^{(k)}, \mu, \nu) ,
\label{tilde_Hd_G}
\\
\tilde{H}^d_{\tilde{G}}(\mu, \nu) &\&\deq \sum_{k=0}^d \sideset{}{'}\sum_{\substack{\mu^{(1)}, \dots, \mu^{(k)} \\ \sum_{i=1}^k \ell^*(\mu^{(i)})= d}}
f_\lambda ({\bf c}) H(\mu^{(1)}, \dots, \mu^{(k)}, \mu, \nu), \\
\tilde{H}^d_G(\mu)&\& := \tilde{H}^d_G(\mu, \nu= (1)^N) , \qquad \tilde{H}^d_{\tilde{G}}(\mu) := \tilde{H}^d_{\tilde{G}}(\mu, \nu= (1)^N). 
\label{tilde_Hd_tildeG}
\eea

\subsection{KP and $2D$-Toda $\tau$-functions of hypergeometric type as generating functions}

Choosing a nonvanishing (small) parameter $\beta$, we define the {\em content product } coefficients
\bea
r_\lambda^{(G, \beta)} &\& \deq   \prod_{(i,j)\in \lambda} r_{j-i}^{G}(\beta),
\label{r_lambda_G}
 \\ 
 r_\lambda^{(\tilde{G}, \beta)} &\& \deq  \prod_{(i,j)\in \lambda} r_{j-i}^{\tilde{G}}(\beta),
\label{r_lambda_tildeG}
\eea
where
\be
r_j^{G}(\beta):= G(j \beta ), \quad r_j^{\tilde{G}}(\beta):= \tilde{G}(j \beta ).
\ee
Introducing a further non vanishing parameter $\gamma$, we define the following parametric family
of $2D$-Toda $\tau$-functions of hypergeometric type (at lattice site $N=0$),
which are shown in \cite{GH1, GH2, H1} to be  generating functions for the corresponding weighted Hurwitz numbers 
\bea
  \tau^{(G, \beta, \gamma)} ({\bf t}, {\bf s}) &\& := \sum_{\lambda} \gamma^{|\lambda|} r^{ (G, \beta)}_\lambda s_\lambda ({\bf t}) s_\lambda ({\bf s})
 = \sum_{d=0}^\infty \beta^d \sum_{\substack{\mu, \nu \\ |\mu|=|\nu|}} \gamma^{|\mu|}  H^d_G(\mu, \nu) p_\mu({\bf t}) p_\nu({\bf s})
\label{tau_G_H}
\\
  \tau^{( \tilde{G},\beta, \gamma) } ({\bf t}, {\bf s}) &\& := \sum_{\lambda} \gamma^{|\lambda|}  r^{(\tilde{G},\beta)}_\lambda s_\lambda ({\bf t}) s_\lambda ({\bf s})
=\sum_{d=0}^\infty \beta^d  \sum_{\substack{\mu, \nu \\ |\mu|=|\nu|}} \gamma^{|\mu|} H^d_{ \tilde{G}}(\mu,  \nu) p_\mu({\bf t}) p_\nu({\bf s}) .
\label{tau_tilde_G_H}
\eea
Here  $s_\lambda({\bf t}), s_\lambda({\bf s})$ are Schur functions,
viewed as functions of two infinite sequences of 2D-Toda flow variables ${\bf t}=(t_1, t_2, \dots)$, ${\bf s}=(s_1, s_2, \dots)$ 
with
\bea
t_i  &\&:= {p_i \over i}, \quad s_i := {p'_i \over i}, \quad i=1, 2, \dots, \\
p_\mu({\bf t}) &\&:= \prod_{i=1}^{\ell(\mu)}p_{\mu_i}  , \quad p_\nu({\bf s}) := \prod_{i=1}^{\ell(\nu)}p'_{\nu_i} 
\eea
interpreted as power sum symmetric functions of a pair of auxiliary sets of variables. 
\br
Note that all six standard bases 
$\{s_\lambda, p_\lambda, m_\lambda, f_\lambda, e_\lambda, h_\lambda\}$ for $\Lambda$ 
introduced in \cite{Mac} appear in the above definitions.
\er

The corresponding generating functions for connected weighted Hurwitz numbers are given, as usual,  by 
 the logarithms
\bea
 \FF^{(G, \beta, \gamma)} ({\bf t}, {\bf s}) &\& := \ln ( \tau^{(G, \beta, \gamma)} ({\bf t}, {\bf s})) 
  = \sum_{d=0}^\infty \beta^d \sum_{\substack{\mu, \nu \\ |\mu|=|\nu|}} \gamma^{|\mu|}  \tilde{H}^d_G(\mu, \nu) p_\mu({\bf t}) p_\nu({\bf s})
\label{logtau_G_H}
\\
  \FF^{( \tilde{G}, \beta, \gamma)} ({\bf t}, {\bf s}) &\& := \ln (\tau^{(G, \beta, \gamma)} ({\bf t}, {\bf s}))  = \sum_{d=0}^\infty \beta^d  \sum_{\substack{\mu, \nu \\ |\mu|=|\nu|}} \gamma^{|\mu|} \tilde{H}^d_{ \tilde{G}}(\mu,  \nu) p_\mu({\bf t}) p_\nu({\bf s}).
\label{logtau_tilde_G_H}
\eea
We may view all these generating functions as elements of  ${\mathbb K}[[{\bf t},{\bf s}, \beta, \gamma]]$.

Below we consider only the case $G$; all formulae for $\tilde{G}$ are analogous. 

\subsection{Summary of results}
\label{summary_results}
In Section \ref{inf_grassmannian} we recall the infinite Grassmannian interpretation of KP and $2D$-Toda $\tau$-functions \cite{Sa, SS, SW, Ta, UTa, Takeb} and the Baker function. For arbitrary elements $W =g(\HH_+)$ of the {\em big cell} of the infinite Grassmannian  $Gr_{\HH_+}(\HH)$, $g \in \GL(\HH)$,
we introduce a dual pair of adapted bases $\{w_k^g(z)\}_{k\in \Zb}$,  $\{w_k^{g*}(z)\}_{k\in \Zb}$,
for which $W = \span\{w_k^g\}_{k\in \Nb^+}$ and the annihilator $W^\perp$ under the Hirota pairing is
$W^\perp = \span\{w_k^{g*}\}_{k\in \Nb^+}$.   

Section \ref{ferm_rep} provides, for general lower triangular group elements $g $, the fermionic vacuum state expectation 
value (VEV) representations of all the relevant quantities:
the $\tau$-functions $\tau_g({\bf t})$, $\tau_g(N, {\bf t }, {\bf s})$, the Baker function $\Psi^-_g(z, N, {\bf t})$ and its dual 
$\Psi^+_g(z, N, {\bf t})$, the adapted bases, the pair and multipair correlators $K^g_{2n}(z_1, \dots, z_n; w_1, \dots, w_n)$,
and the multicurrent correlator $\JJ_n(x_1, \dots , x_n)$. Proposition \ref{wkg_fermionic_prop} gives the VEV  representation of the adapted bases and Proposition \ref{Baker Psi_exp_prop} gives the expansions of $\Psi^-_g(z, {\bf t})$  and $\Psi^+_g(z, {\bf t})$ in these bases. Proposition \ref{zwk_recursion_prop} provides the multiplicative recursion relation satisfied by these bases and  Propositions \ref{zwk_drecursion_prop}, \ref{Pij_prop}  the derivative relations obtained upon application of the Euler operator. Propositions \ref{K_2_tau_prop}, \ref{corollary_2n}  express the  $n$-pair correlator 
$\tilde{K}^g_{2n}(z_1, \dots, z_n; w_1, \dots , w_n)$  in terms of the $\tau$-function and Proposition \ref{K_w_j_series} gives an expansion in binear products of the adapted basis elements.  A Christoffel-Darboux type representation for the pair correlator $K_2^g(z,w)$  is derived in Proposition \ref{CD_rep_prop}.

All this is specialized in Section \ref{Hypersection} to the case of $\tau$-functions of hypergeometric type, which serve as generating functions for weighted Hurwitz numbers. The recursion matrices for the adapted bases are explicitly given in Propositions \ref{Psi_recurs_prop},  \ref{wk_dif_rrs}. The classical and quantum spectral curves that enter in the topological recursion approach \cite{EO1, EO2,  BE, BEMS} are deduced, both directly, using the series expansions of the adapted bases, and using the Kac-Schwarz  operators that generate them from the Baker function at ${\bf t}= {\bf 0}$. The Christoffel-Darboux representation of the pair correlator is given in Proposition \ref{CD_hypergeom_prop} and the pairing matrix is expressed via a generating function in Proposition \ref{generating_function_A}. It is shown in Proposition \ref{multicurrent_Wn_Fn} how the multicurrent generators may be used as alternative generating functions for the weighted Hurwitz numbers.

In Section \ref{caj}, a fermionic derivation is given for the generalized ``cut-and-join'' representation of the $\tau$ function,
and the differential equations it satisfies with respect to the weighting parameters. 

Section \ref{examples} illustrates the results for four specific examples: 1) the original case of simple (single and double) Hurwitz
number studied by Okounkov \cite{Ok} and Pandharipande \cite{Pa}; 2) the case of strongly monotonic Hurwitz numbers or,
equivalently, Belyi curves (with three branch points) studied in \cite{AC1}, \cite{KZ}; 3) the case of weakly monotonic Hurwitz numbers, or signed geometrical enumerative Hurwitz numbers of refs.~\cite{GGN1, GGN2,GH2,HO2,H1} and 4) the {\em quantum Hurwitz numbers} introduced in refs.~\cite{GH2,H1}.

This paper is part of a series \cite{ACEH1}, \cite{ACEH2} on embedding the $\tau$-function  generating function
approach to weighted Hurwitz numbers into the topological recursion scheme. 

\section{Infinite Grassmannian, Baker functions, adapted bases}
\label{inf_grassmannian}

\subsection{The formal  infinite Grassmannian,  $\tau$-functions and Baker functions }
 \label{grassmannian}

In the related papers \cite{ACEH1, ACEH2}, we  use an extension of the Sato \cite{Sa} and Segal-Wilson \cite{SW}
formulations of the infinite dimensional Grassmannian on which the KP flows are understood to act,
replacing the analytic or algebraic definitions of the underlying spaces
and operators by formal series, including extensions by parameters $\{\beta, \gamma, {\bf c}\}$ etc.  
This will be necessary for the proper interpretation of the specialization to $\tau$-functions of hypergeometric type serving as
generators for weighted Hurwitz numbers that appear in Sections \ref{Hypersection} and \ref{caj}.
Our notational conventions derive from reinterpreting  the
elements $W \in \Gr_{\HH_+}(\HH)$ of the Grassmannian  formally, as subspaces
$W\ss \Cb[[z, z^{-1}]]$ of the space of formal Laurent  series in an indeterminate $z$,  with complex coefficients.
These extend what  in \cite{SW} are viewed  as  closed subspaces of the Hilbert space
\be
\HH = L^2(S^1)=  \HH_+ \oplus \HH_-
\ee
of square integrable functions on the unit circle, with orthonormal basis $\{e_i\}_{i\in \Zb }$ 
with $\HH_\pm$  viewed as the spaces of  complex Fourier series
with negative ($-$) and non-negative ($+$) exponentials.
But for the current section and the next, we construct the $\tau$-function, Baker function, adapted bases and associated
correlators in the general setting, without reference to the formal series interpretation. In \cite{Sa, SS} $\HH$  is viewed as a completion
of the space of positive and negative Laurent polynomials
\be
\HH_+ := \overline{\span\{ z^i\}_{i\ge 0}}, \quad \HH_- := \overline{\span\{z^{-i}\}_{i>0}}.
\ee
We use the labelling convention 
\be
\{e_i := z^{-i-1}\}_{i\in \Zb}
\ee
for the orthonormal basis, motivated by the identification of the vacuum element $|0 \rangle$ in the associated Fermonic Fock space as the  Dirac sea,  in which the negative integers denote the filled negative energy states.

In the Segal-Wilson formulation, such subspaces $W\ss \HH$  are  required to be commensurable with $\HH_+ \ss \HH$
in the sense that the orthogonal projection $W\ra \HH$ to $\HH_+$ along $\HH_-$ is required to be 
a Fredholm operator, whose index is called the {\em virtual dimension} of $W$, while the image 
under  orthogonal projection $W \ra \HH_-$ is required, in some sense, to be ``small'' (either Hilbert-Schmidt, as in \cite{SW}, 
or compact). Sato \cite{Sa} uses a different notion, but for our purposes both the functional analytic \cite{SW}
and  differential algebraic \cite{SS, Sa} versions are subsumed by replacing the Hilbert space $\HH$ by the space of formal Laurent
series $\Cb[[z, z^{-1}]]$, extended by the constants $\beta, \gamma$ and ${\bf c}$, viewed as indeterminates.
We nevertheless retain the notation $\HH =\HH_+ + \HH_-$ for the space, and let
\be
(f, g) := \res_{z= 0}\, (z^{-1} f(z) g(1/z))\ \text{``}=\text{''}\ {1\over 2\pi  i}\oint_{z=0}f(z) g({1/z}){dz \over z},
\label{scalar_prod}
\ee
(which might just be a formal residue) denote the complex scalar product in which the basis $\{e_i\}_{i\in \Zb}$ is orthonormal,
with the second equality valid on the smaller space $L^2(S^1)$.

The precise definition of the infinite Lie algebra $\grg\grl(\HH)$ of endomorphisms of $\HH$ and its connected
 Lie group $\GL(\HH)$ of admissible invertible endomorphisms  will be omitted; we refer to \cite{Sa}, \cite{SW} 
 for (different) particular choices, with the differences replaced, for our purposes,  by requiring that their matrix representations
within the orthonormal basis  $\{e_i\}$ and elements in the associative algebra $\Cb[\beta, \gamma, \{g_i\}]$ of formal series 
 in the auxiliary parameters  $(\beta$, $\gamma)$ and the  weighting constants $\{g_i\}_{i\in \Nb^+}$ (or
 ${\bf c}$), when appropriate, should have well defined products and commutators.
  
The two infinite abelian subgroups $\Gamma_{+}$, $\Gamma_{-}$
whose elements consist of multiplication by the exponential series
\be
\Gamma_+= \{\gamma_+({\bf t}) := e^{\xi({\bf t}, z)}\}, \quad \Gamma_- = \{\gamma_-({\bf s}) := e^{\xi({\bf s}, z^{-1})}\},
\ee
where
\be
\xi({\bf t}, z) = \sum_{i=1}^\infty t_i z^i, \quad  \xi({\bf s}, z^{-1}) = \sum_{i=1}^\infty s_i z^{-i}
\ee
are central to the theory since, together with the discrete group
of lattice translations, they  generate the mKP and $2D$-Toda flows.

The KP $\tau$-function associated with an element $W\in \Gr_{\HH_+}(\HH)$ of virtual dimension $N=0$
 is,  within a suitable choice of admissible basis that establishes an isomorphism between the two commensurable
spaces $W$ and  $\HH_+$, the determinant of the orthogonal projection operator to $\HH_+$ from the space 
\be
W({\bf t}) := \gamma_+({\bf t})W
\ee
obtained by applying the abelian group of $\Gamma_+$ flows to $W$
\be
\tau_W({\bf t}) := \det(\pi_+: W({\bf t}) \ra \HH_+).
\label{KP_tau_det}
\ee

It follows \cite{Sa, SS, JM, SW} that this satisfies the Hirota bilinear relations
\be
\oint_{z=\infty} e^{-\xi(\delta{\bf t}, z)} \tau_W({ \bf t}+ \delta{\bf t}  + [z^{-1}]) 
\tau_W({ \bf t} -[z^{-1}] ) dz =0
 \label{hirotaKP} 
\ee
where
\be
 \quad  [z^{-1}]_i  :=  {1 \over  i }z^{-i}.
  \ee
  
The Baker function $\Psi^-_W(x, {\bf t}) $ which, for all values of  the KP flow parameters ${\bf t}$,
takes its values in the annihilator $W^\perp$ of $W$  under the Hirota bilinear pairing
\be
\langle u, \, v \rangle := {1\over 2\pi i} \oint_{z=0} u(z) v(z) dz
\label{hirota_bilinear}
\ee
is then given by the Sato formula
\be
\Psi^-_W(z, {\bf t}) = e^{\xi({\bf t}, z ) }{\tau ({\bf t }- [z^{-1}])\over  \tau({\bf t})}.
\label{Sato_Baker_1}
\ee
 
The dual Baker function, which takes its values in $W$ for all values of ${\bf t}$, is
\be
\Psi^+_W(x, {\bf t}) = e^{-\xi({\bf t}, z ) }{\tau ({\bf t }+ [z^{-1}])\over  \tau({\bf t})}.
\label{Sato_Baker}
\ee

More generally, for $W_N$ of virtual dimension $N$, we similarly define $\tau_{W_N}( {\bf t})$ 
as the determinant of the projections to the subspace 
\be
\HH_+^N := z^{-N} \HH_+,
\label{HHp}
\ee
namely
\be
\tau_{W_N}({\bf t})  := \det(\pi_+:\gamma_+({\bf t})W_N \ra \HH^N_+).
\label{tau_det_W_N}
\ee
In the fermionic  Fock space, their image under the Pl\"ucker map
(see below) is in the charge $N$ sector.

An equivalent view of  the Grassmannian $\Gr_{\HH_+}(\HH)$ is as a quotient  $\GL(\HH)/\GL_{\HH_+}(\HH)$  of
the group $\GL(\HH)$ of admissible invertible endomorphisms of $\HH$ by the 
stabilizer $\GL_{\HH_+}(\HH)$ of the subspace $\HH_+$, which is an infinite dimensional
analog of the maximal parabolic subgroup of  $\GL(N)$ leaving a fixed $k$-dimensional
subspace  invariant. Thus, choosing an element $W \in \Gr_{\HH_+}(\HH)$ within the $\GL(\HH)$ orbit of $\HH_+$
\be
W= g(\HH_+)
\label{Wde}
\ee
 is  equivalent to choosing the left coset 
\be
[g] =g\GL_{\HH_+} (\HH)\in \GL(\HH) / \GL_{\HH_+}(\HH). 
\ee
  If a sequence of elements  of the various virtual dimension components is chosen as
\be
W_N = g(\HH_+^N)
\ee
for a fixed group element $g$, the resulting analogously defined sequence of  KP $\tau$-functions
$\tau_{W_N}({\bf t})$ satisfy the bilinear equations of the mKP hierarchy
\be
 \oint_{z=\infty} z^{N'-N}e^{-\xi(\delta{\bf t}, z)} \tau_{W_N}( { \bf t}+ \delta{\bf t}  + [z^{-1}]) 
\tau_{W_{N'}}({ \bf t} -[z^{-1}] ) dz =0
 \label{hirota_mKP} 
\ee
for $N'\geq N$.

In this case, the group element $g$  is defined within right multiplication by an element of the simultaneous
stabilizer of all the subspaces  $\HH_+^N$, $N\in \Zb$, which is the infinite dimensional analog
of the minimal parabolic (Borel) subgroups of $GL(N)$  consisting of invertible upper triangular matrices. That is,
the mKP flows are interpreted as acting on the codimension $1$ nested sequence of subspaces
\be
\cdots \ss W_{N-1} \ss W_N \ss W_{N+1} \ss \cdots
\label{W_N_flag}
\ee
i.e.,  infinite complete flags  belonging to the $GL(\HH)$ orbit of the standard complete flag 
\be
\cdots \ss \HH_+^{N-1} \ss \HH_+^N \ss \HH_+^{N+1} \ss \cdots.
\ee
The dual infinite flag is given by the orthogonal annihilators
\be
\cdots \ss W^\perp_{N+1} \ss W^\perp_N \ss W^\perp_{N-1} \ss \cdots.
\label{dual_W_N_flag}
\ee
 
\subsection{Adapted bases}
\label{adapted_basis}

By applying the group element $g$ to the standard (monomial) basis  $\{ z^k\}_{k\in \Zb}$ for $\HH$, we obtain an
adapted basis $\{w_k^{g}(z)\}_{k \in \Zb}$ for $\HH$ such that the positively labeled elements $\{w_k^{g}(z)\}_{k \in \Nb^+}$
span the subspace $W=g(\HH_+)$. We also introduce the dual basis $\{w_k^{g*}(z)\}_{k \in \Zb}$ such that
$\{w_k^{g*}(z)\}_{k \in \Nb^+}$ spans the subspace $W^\perp:=\iota\circ (g^t)^{-1}\circ \iota(\HH_+)$ that annihilates
$W$ under the Hirota bilinear pairing (\ref{hirota_bilinear}). Here $\iota$ is the involutive automorphism 
of $\HH$ that relates the Hirota bilinear pairing  to the scalar product  (\ref{scalar_prod}) 
\bea
\iota:\, &\& \HH \ra \HH  \cr
\iota:\, &\&e_i \mapsto e_{-i-1}.
\label{involution}
\eea
Define $\tilde{g}$ to be the conjugation of $g$ under this involution
\be
\tilde{g} := \iota \circ g \circ \iota.
\ee
The corresponding matrix representation of $\tilde{g}$ in the standard basis $e_i$ is then
\be
\tilde{g}_{ij} := g_{-i-1, -j-1}
\ee

Applying the group elements $g$ and $(\tilde{g}^t)^{-1}$ to the standard monomial basis, we obtain the two  bases, 
\bea
w^{g}_k(z) &\&:= g(z^{k-1}) 
= \sum_{j\in \Zb} g_{-j-1, -k} z^j ,  \quad  k \in \Zb,
\label{w_g_def} \\
w^{g*}_k(z) &\& :=({\tilde{g}}^{t})^{-1} (z^{k-1}) 
=\sum_{j\in \Zb} g^{-1}_{k-1, j} z^j, \quad  k \in \Zb,
\label{w_g_star_def}
\eea
with  inverse relations
\bea
z^j =&\& \sum_{k\in \Zb} g^{-1}_{-k, -j-1} w^g_k(z), \quad j \in \Zb
\label{w_g_z_inv}\\
z^j =& \& \sum_{k\in \Zb} g_{j, k-1} w^{g*}_k (z), \quad j \in \Zb.
\label{w_g_dual_z_inv}
\eea
Then 
\be
W := g(\HH_+) = \Span\{w^{g}_k\}_{k\in \Nb^+}, \quad  W^\perp := (\tilde{g}^t)^{-1}(\HH_+) = \Span\{w^{g*}_k\}_{k\in \Nb^+}.
\ee
The bases $\{w_k^g\}_{k\in \Zb}$, $\{w_k^{g*}\}_{k\in \Zb}$ are dual under the Hirota bilinear pairing (\ref{hirota_bilinear})
in the sense that
\be
\langle w_j^g, w^{g*}_l \rangle = \delta_{j + l -1}, 
\ee
 from which it follows that  $W^\perp$ is, indeed, the annihilator of $W$.

We can make the further assumption that the group element $g$ is in
fact represented by lower triangular matrices
\be
g_{ij} = 0 \ {\rm if \ } \ j > i.
\ee
This means that $W \in \Gr_{\HH_+}(\HH)$ must be in the ``big cell'', but imposes no further restriction,
since a representative of the left coset under the stabilizer of $\HH_+$ may always be chosen
in this lower triangular form.

It also implies that the sums (\ref{w_g_def}) - (\ref{w_g_dual_z_inv}) are semi-infinite, with the highest power of $z$  in the Fourier series for $w_k^g(z)$ and $w_k^{g*}(z)$ being $k-1$
\bea
w^{g}_k(z) &\& =  \sum_{j\ = -\infty}^{k-1} g_{-j-1, -k} z^j   \\
w^{g*}_k(z) &\& = \sum_{j= -\infty}^{k-1} g^{-1}_{k-1, j} z^j, 
\label{w_k_g_triangular_series}
   \eea
   and
   \bea
 z^j &\&= \sum_{k=-\infty}^{j+1} g^{-1}_{-k, -j-1} w^g_k(z) \\
  z^j &\&= \sum_{k= -\infty} ^{j+1}g_{j, k-1} w^{g*}_k (z) .
\eea
 
\section{Fermionic representations}
\label{ferm_rep}
 
\subsection{Fermionic Fock space  }
\label{fermionic}

We recall the conventions and notation for the fermionic operator approach to $\tau$-functions
developed by Sato and his school \cite{Sa, SS,JM}. Using the notation of  \autoref{grassmannian}, we denote the semi-infinite wedge product space (Fermonic Fock space)
\be
\FF = \Lambda^{\infty/2} \HH = \oplus_{N\in \Zb} \FF_N
\ee
where $\FF_N$ denotes the charge-$N$ sector, $N\in \Zb$. The pairs $(\lambda, N)$  consisting of a partition $\lambda$
of any weight   (with the convention  that $\lambda_i = 0$ for $i > \ell(\lambda)$) and an integer $N$  label the basis element in $\FF$. These are denoted
\be
|\lambda; N\rangle := e_{\ell_1} \wedge e_{\ell_2} \wedge \cdots  \in \FF_N,
\ee
where
\be
\ell_i := \lambda_i -i +N 
\ee
is a strictly decreasing sequence of integers saturating, after a finite number of terms, in consecutively
decreasing integers. These are viewed as the integer lattice points of occupied sites,  and the vacuum element in each sector is given by the trivial partition
\be
| N \rangle := |0; N\rangle = e_{N-1} \wedge e_{N-2} \wedge \cdots.
\ee

The Fermi creation and annihilation operators $\{\psi_i, \psi^*_i\}_{i\in \Zb}$ are the generators of
 the Fermi Fock space representation of the  Clifford algebra on $\HH \oplus \HH^*$ with 
quadratic form 
\be
\Qb(X +\alpha):= 2\alpha(X)  \  \text{for } X \in \HH, \  \alpha \in \HH^* .
\ee
In the standard irreducible infinite Clifford module $\FF$, these  are given by the linear action of 
exterior and interior multiplication by the basis elements $\{e_i\}_{i\in \Zb}$ for $\HH$ and their duals $\{\tilde{e}^i\}_{i\in \Zb}$ for $\HH^*$
\be
\psi_ i := e_i  \wedge,  \quad \psi^*_i := i_{\tilde{e}^i}  \in \End(\FF).
\ee
This implies the generating relations of the Clifford algebra 
\be
[\psi_i, \psi^*_j]_+ := \psi_i \psi^*_j +\psi^*_j \psi_i=\delta_{ij}, \
\quad [\psi_i, \psi_j]_+ =[\psi^*_i, \psi^*_j] _+= 0.
\ee

Normal ordering  $\no{\OO}$ of any finite product of such creation and
annihilation operators is defined by placing all the annihilation factors $\psi^*_i$ to the right  for $i \ge 0$
and all the creation factors $\psi_i$ to the left for $i < 0$ (taking the signs following from anticommutativity
into account). This implies the vanishing of  vacuum expectation values  (VEV's) of all normally ordered operators 
(other than scalars)
\be
\langle 0 | \no{\OO} | 0 \rangle= 0.
\ee

If the partition $\lambda$ is expressed in Frobenius notation as $(a_1, \dots a_r | b_1, \dots , b_r)$, where the $a_i$'s
and $b_i$'s are strictly decreasing sequences of of non-negative integers, denoting the number of boxes to the 
right  of and below  the $i$th diagonal box of the corresponding Young diagram (the ``arm'' and ``leg'' lengths),
with the Frobenius index $r$ equal to the number of terms on the principal diagonal, we may express the 
basis element $|\lambda; N\rangle$ as
\be
|\lambda; N\rangle = (-1)^{\sum_{i=1}^r b_i} \prod_{i=1}^r \left(\psi_{a_i +N} \psi^*_{-b_i -1 +N}\right) |N \rangle.
\label{basis_states}
\ee

The Fermi fields $\psi(z)$ and $\psi^*(z)$ are defined as generating series for these creation
and annihilation operators
\be 
\psi(z) := \sum_{i\in \Zb}\psi_i z^i,  \quad \psi^*(z) := \sum_{i\in \Zb} \psi_i^*  z^{-i-1}.
\label{psi_series}
\ee 

Elements $g \in \GL(\HH) $ of the  connected group of invertible endomorphisms of $\HH$
may be viewed as elements of the larger group $\Spin (\HH+\HH^*, \Qb)$ of endomorphisms
of $\HH+\HH^*$ preserving the quadratic form $\Qb$. If they are represented  within
the basis $\{e_i\}_{i \in \Zb}$ by the doubly infinite matrix exponential
\be
g \sim e^A
\ee
with elements $\{A_{ij}\}_{i,j \in \Zb}$, they have the Clifford representation
\be
\hat{g} := \exp{\sum_{i, j \in \Zb} A_{ij} 
\no{\psi_i \psi^*_j}}.
\ee

 The Pl\"ucker map into the projectivization $\Pb\FF$ of the Fock space
 \bea
 \grP\grl : \Gr_{\HH_+}(\HH) &\& \ra\Pb\FF\cr
 \grP\grl : W:= \span\{w_1, w_2, \dots \}&\& \mapsto [w_1 \wedge w_2 \wedge \cdots ]
 \eea
 (where $[ \dots ]$ denotes the projective class) embeds the infinite Grassmannian $\Gr_{\HH_+}(\HH)$ 
 into $\Pb\FF$ as the  union of the orbits of the subspaces $ \HH_+^N$
 under the Clifford action of $\GL(\HH)$  if $W$ has virtual dimension $N$.
 Thus the Pl\"ucker map  intertwines the action of $\GL(\HH)$ on $Gr_{\HH_+}(\HH)$
and the (projectivized) Clifford representation on $\FF$
 \be
\grP\grl( g(W) ) = \hat{g} (\grP\grl(W)), \quad \forall \ W \in \Gr_{\HH_+}(\HH).
  \ee
  
  In particular, the Clifford representation of the KP  and 2D Toda flow groups  $\Gamma_+$ and
  $\Gamma_-$ is given by
  \be
\hat{\gamma}_+({\bf t} )= e^{\sum_{i=1}^\infty t_i J_i}, \quad \hat{\gamma}_-({\bf s}) = e^{\sum_{i=1}^\infty s_i J_{-i}}, 
\ee
where the operators $J_i$  are bilinears in the creation and annihilation operators 
\be
J_i := \sum_{j\in \Zb} \no{\psi_j \psi^*_{j+i}}, \quad i\in \Zb
\ee
that satisfy the commutation relations of the infinite Heisenberg algebra
\be
[J_i, J_{-j }] = i \delta_{ij} .
\ee
The projective representation may be viewed as a representation of this  central extension 
of the abelian group generated by the elements $\gamma_+({\bf t} )$, $\gamma_-({\bf s)}$,
with
\be
\hat{\gamma}_+({\bf t} )\hat{\gamma}_-({\bf s} )=  e^{\sum_{i=1}^\infty i t_i s_i }\, \hat{\gamma}_-({\bf s}) \hat{\gamma}_+({\bf t}).
\ee
The $J_i$'s  may be viewed as Fourier components of the current operator
\be
J(z) := \sum_{i\in \Zb} J_i z^i. 
\ee

The effect of conjugation of the creation and annihilation operators by a group element is
\be
\hat{g}\psi_i \hat{g}^{-1} = \sum_{j \in \Zb} \psi_j g_{ji},  \quad
 \hat{g}^{-1}\psi^*_i \hat{g} = \sum_{j \in \Zb} g_{ij} \psi^*_j 
 \label{g_psi_z_conj}
 \ee
 if
 \be
 g(e_i) =\sum_{j\in \Zb} e_j g_{ji}.
 \ee

In particular,
\bea
 \label{fermiconj}
\hat{\gamma}_+({\bf t})  \psi(z) \hat{\gamma}^{-1} _+ ({\bf t})&\& = e^{\xi ({\bf t}, z)} \psi(z), \quad
 \hat{\gamma}_+({\bf t})  \psi^*(z) \hat{\gamma}^{-1} _+ ({\bf t})= e^{-\xi ({\bf t}, z)} \psi^*(z), 
 \label{gamma+psi_conj}\\
&\& \cr
 \hat{\gamma}_-({\bf s})  \psi(z) \hat{\gamma}^{-1} _- ({\bf s})&\& = e^{\xi ( {\bf s}, z^{-1})} \psi(z), \quad
 \hat{\gamma}_-({\bf s})  \psi^*(z) \hat{\gamma}^{-1} _-({\bf s})= e^{-\xi ( {\bf s}, z^{-1})} \psi^*(z). 
  \label{gamma-psi_conj}
\eea

\subsection{Fermionic construction of $\tau$-functions, Baker function \\ and adapted bases }
\label{fermionic_tau}

In Sato's approach to KP $\tau$-functions \cite{Sa, SS}, given an element $\hat{g}$ of  the Fermionic
representation of the Clifford algebra that satisfies the bilinear relation.
\be
[\II, \hat{g}\otimes  \hat{g} ]=0,
\label{bilinear_g}
\ee
where 
\be
\II := \sum_{i\in \Zb} \psi_i \otimes \psi^*_i  \in \End(\FF \otimes \FF),
\ee
the following vacuum expectation value
\be
\tau_g({\bf t}) = \langle 0 | \hat{\gamma}_+({\bf t}) \hat{g}|0 \rangle
\label{tau_g_VEV}
\ee
 satisfies the Hirota bilinear relation (\ref{hirotaKP}) and hence is a KP $\tau$-function.
   In the following, we call such elements {\it group-like elements} \cite{AZ}.
 In particular, if
\be\label{fermigroup}
\hat{g} = e^{\sum_{i,j \in \Zb} A_{ij }: \psi_i \psi^*_j :}
\ee
is the fermionic representation of the element  $g\in \GL(\HH)$ with matrix representation
$e^A$ in the standard $\{e_i\}_{i\in \Zb}$ basis, $\hat{g}$ satisfies (\ref{bilinear_g}) and $\tau_g({\bf t}) 
$ is the same as the $\tau$-function
$\tau_W({\bf t}) $ defined in (\ref{KP_tau_det}) with $W=g\HH_+$.
The Baker function defined in (\ref{Sato_Baker_1}) can then be expressed as the following VEV \cite{Sa, SS}
\be
\Psi^-_g(z, {\bf t}) := \Psi^-_W(z, {\bf t}) = {\langle 0 | \psi^*_0\hat{\gamma}_+({\bf t} )\psi(z) \hat{g} | 0\rangle \over \tau_g({\bf t}) }, 
\label{Baker_fermionic}
\ee
and the dual Baker function (\ref{Sato_Baker}) is
\be
\Psi^+_g(z, {\bf t}):=\Psi^+_W(z, {\bf t}) = {\langle 0 | \psi_{-1}\hat{\gamma}_+({\bf t} )\psi^*(z) \hat{g} | 0\rangle \over \tau_g({\bf t}) }. 
\label{dual_Baker_fermionic}
\ee

Since
\be
\langle 0 | \hat{\gamma}_+({\bf t}) | \lambda \rangle  = s_\lambda({\bf t}) = \langle \lambda | \hat{ \gamma}_-({\bf t}) | 0 \rangle,
\ee
it follows that the $\tau$-function may, at least formally, be expressed as a sum over the Schur function basis
\be
\tau_g({\bf t}) = \sum_{\lambda} \pi_\lambda(g) s_\lambda({\bf t})
\label{tau_schur_exp}
\ee
where 
\be
\pi_\lambda(g)  := \langle \lambda |\hat{g} | 0 \rangle
\ee
is the $\lambda$ Pl\"ucker coordinate of the Grassmannian element $W$
and the Schur functions $s_\lambda$ are viewed as functions of
the  KP flow variables ${\bf t} = (t_1, t_2, \dots)$.

More generally, for any such $\hat{g}$, we may define a lattice of KP $\tau$-functions
generating solutions to the modified KP (mKP) hierarchy, 
\be
\tau^{mKP}_g(N, {\bf t}) := \langle N | \hat{\gamma}_+({\bf t}) \hat{g}|N \rangle
\ee
which  satisfy the extended system of Hirota equations  \eqref{hirota_mKP} of the mKP hierarchy.

Introducing a second set of flow variables ${\bf s}=(s_1, s_2, \dots)$ we may define the
$2D$ Toda lattice of $\tau$-functions
\be
\tau_g(N, {\bf t}, {\bf s}) := \langle N | \hat{\gamma}_+({\bf t}) \hat{g} \hat{\gamma}_-({\bf s})|N \rangle,
\ee
which  satisfy
\bea
 &\&  \oint_{z=\infty}  z^{N'-N}e^{-\xi(\delta{\bf t}, z)} \tau_g(N,{ \bf t} +[z^{-1}] , {\bf s})
 \tau_g(N', { \bf t}+ \delta{\bf t}  - [z^{-1}],  {\bf s} + \delta{\bf s}) dz = 
 \cr
 &\&  \oint_{z=0}  z^{N'-N} e^{-\xi(\delta{\bf s}, z^{-1})} \tau_g(N-1, {\bf t},   { \bf s} +[z])
 \tau_g(N'+1, {\bf t} + \delta{\bf t}, {\bf s}+ \delta{\bf s}  - [z] ) dz,
\label{hirota2d_Toda}
 \eea
understood to hold  identically in  $\delta{\bf t}   = (\delta t_1, \delta t_2, \dots),\   \delta{\bf s} := (\delta s_1, \delta s_2, \dots)$ and in $N$, $N'$.

In the following, we  make the further assumption that the group element $g$ is  represented by lower triangular matrices
\be
g_{ij} = 0  \quad \text{if} \ \  j > i.
\ee
This means that
\be
\langle 0 | \hat{g} = \langle 0 | \hat{g}^{-1} = \langle 0 |.
\ee

Applying the group element $\hat{g}$, we obtain the fermionic VEV respresentation of the
bases $\{ w^{g}_{k \in \Zb}\}$ , $\{ w^{g*}_{k \in \Zb}\}$  defined in (\ref{w_g_def}), (\ref{w_g_star_def}).
\begin{proposition}
\label{wkg_fermionic_prop}
The bases defined in eqs.~(\ref{w_g_def}) , (\ref{w_g_star_def}) have the following representation as fermionic VEV's
\bea
w_k^g(z) &\& =\begin{cases} \langle 0 | \psi_{-k} \hat{g}^{-1} \psi^*(z) \hat{g} | 0\rangle \quad  \text{if} \ k\ge 1, \cr
\langle 0 |\hat{g}^{-1} \psi^*(z) \hat{g} \psi_{-k} | 0\rangle  \quad \text{if} \ k \le  0, 
\end{cases}
\label{wkg_fermionic}
\\
w_k^{g*}(z) &\& =  \begin{cases} \langle 0 |\psi^*_{k-1} \hat{g}^{-1} \psi (z)\hat{g}  |0 \rangle  \quad  \text{if} \ k \ge 1, \cr
\langle 0 | \hat{g}^{-1} \psi(z) \hat{g} \psi^*_{k-1}  |0 \rangle \quad  \text{if} \ k\le 0.
\end{cases} 
\label{wkg*_fermionic}
\eea
\end{proposition}
\begin{proof}
Substitute the fourier series representations (\ref{psi_series}) of the fermionic field operators $\psi(z)$, $\psi^*(z)$
in the RHS of (\ref{wkg_fermionic}) and (\ref{wkg*_fermionic}) and apply the conjugation relations (\ref{g_psi_z_conj})
term by term. This results in the series representations (\ref{w_g_def}), (\ref{w_g_star_def}) for 
$\{ w^{g}_{k \in \Zb}\}$ , $\{ w^{g*}_{k \in \Zb}\}$, both for $k\ge 1$ and $k \leq 0$.
\end{proof}
Now define the sequence of subspaces $W_N^{ g}, W_N^{ g*}$ , $N\in \Zb$
\bea\label{WchargeN}
W_N^{ g} &\&= \overline{\span\{w_{-N+k}^g\}_{k\in \Nb^+}} \in \Gr_{\HH_+^N}(\HH)\\
W_N^{ g\perp} &\&= \overline{\span\{w_{N+k }^{g*}\}_{k\in \Nb^+}}  \in \Gr_{\HH_+^{-N}}(\HH),
\eea
which coincide with (\ref{W_N_flag}), (\ref{dual_W_N_flag}).

The  charge $N$ sector Baker function and  its dual 
\bea
\Psi^-_g(z, N, {\bf t}) &:=&  {\langle N | \psi^*_N\hat{\gamma}_+({\bf t} )\psi(z) \hat{g} | N\rangle \over \tau_g({\bf t}) } 
 \in W_N^{ g\perp} 
  \\
\Psi^+_g(z, N,  {\bf t})&:=& {\langle N | \psi_{N-1}\hat{\gamma}_+({\bf t} )\psi^*(z) \hat{g} | N \rangle \over \tau_g({\bf t}) } \in W_N^{g} 
\eea
 can be expanded in these bases. For the case $N=0$, we have
 \begin{proposition}
 \label{Baker Psi_exp_prop}
 \bea
 \Psi^-_g(z,  {\bf t}) &\&= \sum_{j=1}^\infty \alpha^{g}_j({\bf t}) w^{g*}_{j}(z) 
 \label{Psi_w_k_g_exp}
 \\
 \Psi^+_g(z,  {\bf t}) &\& = \sum_{j=1}^\infty \beta^{g}_j({\bf t}) w^{g}_{j}(z) ,
  \label{Psi*_w_k_g_exp}
 \eea
 where
 \be 
\alpha^{g}_j({\bf t}) := {\langle 0 | \psi^*_{0}\hat{ \gamma}_+({\bf t}) \hat{g}\ \psi_{j-1}|0 \rangle \over \tau_g(N, {\bf t})},
 \quad  \beta^{g}_j({\bf t}) := { \langle 0 | \psi_{-1}\hat{ \gamma}_+({\bf t}) \hat{g} \psi^*_{-j} |0 \rangle \over \tau_g(N, {\bf t})}.
 \ee
  \end{proposition}
  \begin{proof}
  In eq.~(\ref{Baker_fermionic}), insert $\Ib= \hat{g} \hat{g}^{-1}$ and the projection operator  to the $N=1$
  sector $\FF_1 \ss \FF$ of the fermionic Fock space
   \be
\Pi_1 := \sum_{\lambda} |\lambda; 1\rangle \langle \lambda ; 1 | ,
 \ee
  (where the sum is over all partitions $\lambda$). This gives
  \be
  \Psi^-_g(z, {\bf t}) = {1\over \tau_g({\bf t})}\sum_{\lambda} \langle 0 | \psi_0^* \hat{\gamma}_+({\bf t}) \hat{g} | \lambda; 1\rangle 
  \langle \lambda; 1| \hat{g}^{-1} \psi(z) \hat{g} | 0 \rangle.
  \label{sum_intermediate_states_baker}
  \ee
 In this sum, the only nonvanishing terms are the trivial partition $\lambda = \emptyset$ and the horizontal
 partitions $\lambda = (j), \ j=1 , 2,  \dots$. To see this, use (\ref{g_psi_z_conj}) to write the second factor 
 in the sum (\ref{sum_intermediate_states_baker}) as
 \be
 \langle \lambda ; 1 | \hat{g}^{-1} \psi(z) \hat{g} | 0 \rangle =
 \sum_{j\in \Zb} z^j \sum_{i=j}^\infty  \langle \lambda ; 1 | \psi_i | 0 \rangle  g^{-1}_{ij}.
 \ee
Expressing the partition $\lambda$  in Frobenius notation as $(a_1, \dots, a_r | b_1, \dots , b_r)$,
 the state $|\lambda ; 1\rangle$ is, by (\ref{basis_states})
  \be
 |\lambda; 1 \rangle = (-1)^{\sum_{i=1}^r b_i} \prod_{i=1}^r \left(\psi_{a_i +1} \psi^*_{-b_i} \right)\psi_0 |0 \rangle.
 \ee
 It follows from the fermionic Wick theorem that
 \be
 \langle \lambda; 1| \psi_i | 0 \rangle = \langle 0 | \psi_0^* \prod_{i=1}^r \left(\psi_{-b_i} \psi^*_{a_i +1} \right) \psi_i | 0 \rangle
 \ee
 vanishes unless either $\lambda = \emptyset$ or $r=1$ and $b_1=0$. For the first case, we have
 \be
 \langle \emptyset; 1| \hat{g}^{-1}\psi(z)\hat{g} |0 \rangle = w_1^{g*}(z),
 \ee
 for the other case $\lambda = (j)$, we have
 \be
 \langle (j); 1 |\hat{g}^{-1} \psi(z) \hat{g} | 0 \rangle = w_{j+1}^{g*}(z) , \quad j=1, 2, \dots
 \ee 
 by eq.~(\ref{w_g_star_def}), and
 \be
 \langle 0 | \psistar_0 \hat{\gamma}_+({\bf t}) \hat{g} | (j-1); 1\rangle = \tau_g({\bf t})  \alpha_j^g({\bf t}),
 \ee
 proving eq.(\ref{Psi_w_k_g_exp}). Eq.~(\ref{Psi*_w_k_g_exp}) is proved similarly.
  \end{proof}
     
\subsection{Multiplicative recursion relations for general $\hat{g}$}
\label{multipl_rrs}
     
   Multiplying the adapted basis elements by $z$, it follows that they satisfy the  recursion relations
   \begin{proposition}
   \label{zwk_recursion_prop}
\bea
z w_k^g(z) &\&= \sum_{j=-\infty}^{k+1} \tilde{Q}^+_{kj} w_j^g(z) 
\label{zwk_recursion}
\\
z w_k^{g*}(z) &\&= \sum_{j=-\infty}^{k+1} \tilde{Q}^-_{kj} w_j^g(z),
\label{zwk_recursion_dual}
\eea
where
\bea
\tilde{Q}^+_{kj} &\&= \sum_{i=-k-1}^{-j}g^{-1}_{-j, i} g_{i+1, -k}  = \tilde{Q}^-_{1-j, 1-k }, \quad j\le k+1
\label{tilde_Q_+}\\
\tilde{Q}^-_{kj} &\&= \sum_{i=j-2}^{k-1}g^{-1}_{k-1, i} g_{i+1, j-1}  = \tilde{Q}^+_{1-j, 1-k},  \quad j\le k+1.
\label{tilde_Q_-}
\eea
\end{proposition}

\begin{proof}
This can either be proved directly, by inverting the Fourier series expansions for $w_k^g(z)$ and $w_k^{g*}(z)$, 
or by using the fermionic identities
\be
z \psi(z) = [J_1, \psi(z)], \quad z \psi^*(z) = - [J_1, \psi^*(z)]
\ee
and evaluating the RHS of the fermionic representations  (\ref{wkg_fermionic}) and (\ref{wkg*_fermionic}) for
both $w_k^g(z)$ and $w_k^{g*}(z)$ by introducing a sum over a complete set of intermediate states in the $N=0$ or $N=-1$ fermionic charge sectors; i.e., by inserting the projection operator 
\be
\Pi_N:= \sum_\lambda  |\lambda; N\rangle \langle \lambda; N |
\ee
onto these sectors.
\end{proof}

The recursion matrix $\tilde{Q}^+$ may be split into block form
\be
\tilde{Q}^+ = \begin{pmatrix} \tilde{Q}^{--} &  \vdots &\tilde{Q}^{-+}  \\ 
\cdots  & & \cdots \\ \tilde{Q}^{+-} & \vdots & \tilde{Q}^{++} 
    \end{pmatrix},
\ee
where the range of indices are:
\bea
 \tilde{Q}^{--}_{kl}:&\&  \ k\le 0, \ l\le 0; \quad    \tilde{Q}^{-+}_{kl}: k\le 0, \ l\ge 1; \cr
  \tilde{Q}^{+-}_{kl}: &\& \ k\ge 1, \ l\le 0; \quad   \tilde{Q}^{++}_{kl}: k\ge 1, \ l\ge 1.
  \quad 
\eea
It follows from the above computation that the fermionic representation of the recursion  matrix elements is 
\begin{corollary}
\label{fermi_rep_Q}
\bea
 \tilde{Q}^{--}_{kl} &\& =   \langle 0 | \psi^*_{-l} \hat{g}^{-1} J_1 \hat{g} \psi_{-k} | 0 \rangle 
 - \delta_{kl} \kappa_+, \quad k\le 0, \ l\le 0 \cr
  \tilde{Q}^{-+}_{kl} &\& =   \delta_{k0}\delta_{l 1} g_{00}g^{-1}_{-1, -1},  \quad  k\le 0, \ l\ge 1\cr
  \tilde{Q}^{+-}_{kl} &\& =  \langle 0 | \psi_{-k} \psi^*_{-l}\hat{g}^{-1} J_1 \hat{g}  | 0 \rangle, 
  \quad k\ge 1, \ l\le 0 
 \cr  
      \tilde{Q}^{++}_{kl} &\& =  - \langle 0 | \psi_{-k} \hat{g}^{-1} J_1 \hat{g} \psi^*_{-l} | 0 \rangle 
  + \delta_{kl} \kappa_-, \quad k\ge 1, \ l\ge 1 ,
\eea
where
\be
\kappa_+ := \sum_{j=0}^{\infty} (g_{j+1, j} g^{-1}_{j j} + g_{jj} g^{-1}_{j, j-1} ),\quad
\kappa_- := \sum_{j= -\infty}^{-1}(g_{j+1, j} g^{-1}_{jj} + g_{jj} g^{-1}_{j, j-1} ).
\ee
\end{corollary}
  Note that   $ \tilde{Q}^{-+}_{kl}$ has just one nonzero entry, in the lower triangular corner $(01)$,
  and both $ \tilde{Q}^{--}_{kl} $ and  $ \tilde{Q}^{++}_{kl} $ are nearly lower triangular, with only  a single
  nonzero diagonal immediately above the principal one.
  
  \begin{remark}
  Proposition \ref{zwk_recursion_prop} provides a system of recursion relations satisfied by the adapted
  basis elements $\{w_k^g(z)\}$,  $\{w_k^{g*}(z)\}$ in the sense that the recursion matrices $\tilde{Q}^\pm$
  are almost lower triangular, with a single nonvanishing diagonal above the principle one. This means that,
  given all the preceding $\{w^g_j(z)\}$'s, for $j\le k$, we can uniquely determine  $w^g_{k+1}(z)$
  (and similarly for $w^{g*}_{k+1}(z)$). This is only a finite recursion system in the case when the matrices
  $\tilde{Q}^\pm$ are of finite band width below the principle diagonal. In Section \ref{Hypersection} it will
  be shown that this is indeed the case when we specialize to group elements of the convolution type,
  with a generating function $G(z)$ that is polynomial, and the number of nonvanishing parameters $s_i$
  in the abelian group element $\hat{\gamma}_-({\bf s})$ finite.
  \end{remark}
 
\subsection{Euler derivative relations for general $\hat{g}$}
\label{deriv_rec_rels}

Applying the Euler operator 
\be\label{Eulerop}
\DD:= z {d \over dz}
\ee
to the adapted bases $ \{w_k^g(z)\}$, $ \{w_k^{g*}(z)\}$
gives the following linear relations
\begin{proposition}
\label{zwk_drecursion_prop}
\bea
\DD w_k^g(z) &\&= \sum_{j=-\infty}^{k} \tilde{P}^+_{kj} w_j^g(z)
\label{zwk_drecursion}
\\
\DD w_k^{g*}(z) &\&=  \sum_{j=-\infty}^{k} \tilde{P}^-_{kj} w_j^{g*}(z) 
\label{zwk_drecursion_dual}
\eea
where the lower triangular matrices $\tilde{P}^{\pm} $ have as entries
\bea
\tilde{P}^+_{kj} &\&= \sum_{i=j}^{k} i g^{-1}_{-j, -i} g_{-i, -k}   - \delta_{j k}, \quad j\le k, \\
\tilde{P}^-_{kj} &\&= \sum_{i=j}^{k}  i g^{-1}_{k-1, i-1} g_{i-1, j-1} - \delta_{jk}  \quad j\le k.
\label{tilde_P_pm}
\eea
\end{proposition}
\begin{proof}
This can again be proved either directly, by inverting the Fourier series expansions for $w_k^g(z)$ and $w_k^{g*}(z)$, 
or by using the fermionic identities
\be
\DD \psi(z) = [\hat{H}, \psi(z)], \quad \DD \psi^*(z) = - [\hat{H}, \psi^*(z)] -\psi^*(z),
\ee
where $\hat{H}$ is the fermionic energy operator
\be
\hat{H} := \sum_{j\in \Zb} j \,\no { \psi_j \psi^*_j }, 
\ee
and again evaluating the RHS for the fermionic expressions by introducing the projector $\Pi_N$
to the $N=0$ or $N=-1$ fermionic charge sectors.
\end{proof}

The differential recursion matrix $\tilde{P}^+$ may again be split into block form
\be
\tilde{P}^+ = \begin{pmatrix} \tilde{P}^{--} &  \vdots &\tilde{P}^{-+}  = {\bf 0} \\ 
\cdots  & & \cdots \\ \tilde{P}^{+-} & \vdots & \tilde{P}^{++} 
    \end{pmatrix}
\ee
where the range of indices are again:
\bea
 \tilde{P}^{--}_{kl}:&\&  \ k\le 0, \ l\le 0; \quad    \tilde{P}^{-+}_{kl}: k\le 0, \ l\ge 1; \cr
  \tilde{P}^{+-}_{kl}: &\& \ k\ge 1, \ l\le 0; \quad   \tilde{P}^{++}_{kl}: k\ge 1, \ l\ge 1.
   \quad 
\eea

\begin{proposition}
\label{Pij_prop}
The fermionic representation of the Euler differential  matrix elements is
\bea
 \tilde{P}^{--}_{kl} &\& =  - \langle 0 | \psi^*_{-l} \hat{g}^{-1} \hat{H} \hat{g} \psi_{-k} | 0 \rangle , \quad k \le 0, \ l\le 0 \cr
  \tilde{P}^{-+}_{kl} &\& =   0,  \quad  k\le 0, \ l\ge 1\cr
  \tilde{P}^{+-}_{kl} &\& =  \langle 0 | \psi_{-l} \psi^*_{k}\hat{g}^{-1} \hat{H}\hat{g}  | 0 \rangle -  \delta_{k, -l},  
  \quad k\ge 1, \ l\le 0 
 \cr  
      \tilde{P}^{++}_{kl} &\& =  - \langle 0 | \psi_{-k} \hat{g}^{-1} \hat{H} \hat{g} \psi^*_{-l} | 0 \rangle ,
      \quad k \ge 1, \ l \ge 1.
\eea
  Note that   both $ \tilde{P}^{--}_{kl} $ and  $ \tilde{P}^{++}_{kl} $ are lower triangular.  
  \end{proposition}
  \begin{proof}
  A direct evaluation of the matrix elements.
  \end{proof}
  
  \begin{remark}
  Proposition \ref{Pij_prop} provides a recursive system of linear differential relations satisfied by the adapted
  basis elements $\{w_k^g(z)\}$,  $\{w_k^{g*}(z)\}$ in the sense that the  matrices $\tilde{P}^\pm$
  are lower triangular. This means that, given all the $\{w^g_j(z)\}$'s, for $j\le k$, 
  we can uniquely determine $\DD w^g_k(z)$   (and similarly for $\DD w^{g*}_k(z)$).
  This is again only a finite system in the case when the matrices
  $\tilde{P}^\pm$ are of finite band width below the principal diagonal. In Section \ref{Hypersection} it will
  be shown that this is indeed the case when we require that  the number of nonvanishing parameters $s_i$
  in the abelian group element $\hat{\gamma}_-({\bf s})$ be finite.
  \end{remark}
 
\subsection{The $n$-pair correlation function  and Christoffel-Darboux type kernel}
\label{pair_correlators}
 
 The $n$-pair correlation function is defined by the following VEV
 \be
 \tilde{K}^g_{2n}(z_1, \dots , z_n; w_1, \dots , w_n) := 
 \langle 0| \prod_{i=1}^n \left( \psi(z_i)\psi^*(w_i) \right) \hat{g} |0 \rangle ,
 \label{K_2n}
 \ee
 provided $|z_i | > |w_i| \ \forall i$.
 Adding the KP multitime flow variables gives the time dependent correlation function
 \be
\tilde{K }^g_{2n}(z_1, \dots , z_n; w_1, \dots , w_n, {\bf t}) := 
 \langle 0|\hat{\gamma}_+({\bf t})  \prod_{i=1}^n \left( \psi(z_i)\psi^*(w_i) \right)\hat{g} |0 \rangle .
 \label{K_2n_time}
 \ee
 In particular, for $n=1$, ${\bf t}= {\bf 0}$,
 \be
\tilde{K}^g_{2}(z; w) := \begin{cases}
 \langle 0| \psi(z)\psi^*(w)  \hat{g} |0 \rangle \quad \text{if} \quad |z| > |w|.
  \\
- \langle 0|  \psi^*(w)\psi(z) \hat{g} |0 \rangle \quad \text{if} \quad |z| < |w|.
 \label{K2_VEV_rep}
 \end{cases}
 \ee
 This is equivalent to the following expression in terms of the $\tau$-function
 \begin{proposition}
 \label{K_2_tau_prop}
\be
\tilde{K}^g_2(z,w) = { \tau_g( [w^{-1}]-[z^{-1}]) \over z-w}. 
\label{K_2_tau}
\ee
\end{proposition}
\begin{proof}
See Appendix \ref{fermionic_VEV_K2}
\end{proof} 
 
 By Wick's theorem, $\tilde{K}^g_{2n}$  may be expressed as a determinant in terms of the $2$-point function
 \begin{lemma}
 \be
\tilde{K}^g_{2n}(z_1, \dots , z_n; w_1, \dots , w_n, {\bf t}) =\det(\tilde{K}^g_2(z_i, w_j, {\bf t})|_{1 \le i,j \le n}
  \label{K_2n_time_lemma}
 \ee
\end{lemma}
To express the multipair correlator $\tilde{K}^g_{2n}$ in terms of the $\tau$-function it is helpful
to introduce the following notation.
\begin{definition}
For a set of  $n$ nonzero complex numbers ${\bf z} = (z_1, z_2, \dots z_n )$,  we use the
  notation $[{\bf z}^{-1}]$ to denote the infinite sequence $\{[{\bf z}]_i\}_{i=1, 2, \cdots}$, 
  of normalized power sums
\be
[{\bf z}^{-1}] := ( [{\bf z}^{-1}]_1,  [{\bf z}^{-1}]_2, \dots ), \quad  [{\bf z}^{-1}]_i :=  {1\over i} \sum_{j=1}^n{1\over( z_j)^i}.
  \ee
\end{definition}

  \begin{lemma}
  \label{main_fermionic_identity}
  For $2n$ distinct nonzero complex parameters $\{z_i, w_i\}_{i=1, \dots n}$ provided $|z_i | > |w_i| \ \forall i$, we have
  \be
  \langle 0 | \prod_{i=1}^n \left(\psi(z_i) \psi^*(w_i)\right) =
  \det \left({1\over z_i - w_j} \right)_{1\le i,j \le n}\langle 0 | \hat{\gamma}_+([{\bf w}^{-1}] - [{\bf z}^{-1}]).
  \label{psi_psi_dag_gamma_+}
  \ee
  \end{lemma}
  
  \begin{proof}
 See Appendix \ref{K_g_2n_tau}.
  \end{proof}
  
It follows that, for an arbitrary group element $g \in \GL(\HH)$, we can express the $2n$-point time dependent 
  pair correlation  function (\ref{K2_VEV_rep}) in terms of the $\tau$-function.
  \begin{corollary}
  \label{corollary_2n}
  \bea
  \tilde{K}^g_{2n}( {\bf z}, {\bf w},{\bf t})  = \prod_{j=1}^n( e^{\xi({\bf t,} z_j)-\xi({\bf t}, w_j)} )
  \tau_g({\bf t} -[{\bf z}^{-1}] + [{\bf w}^{-1}]) \det_{i,j=1}^n \left({1\over z_i - w_j} \right) 
  \eea
    \end{corollary}
       \begin{proof}  (See Appendix \ref{K_g_2n_tau}.) It is first proved for $n=1$, by applying Lemma \ref{main_fermionic_identity} for $n=1$
       to deduce
       \bea
       \langle 0 | \hat{\gamma}_+({\bf t} ) \psi(z)\psi^*(w)\hat{g} | 0 \rangle 
       &=& e^{\xi({\bf t}, z) - \xi({\bf t}, w)} \langle 0 |\psi(z) \psi^*(w) \hat{\gamma}_+({\bf t}) \hat{g}|0 \rangle 
      \cr
       &=&  {e^{\xi({\bf t}, z) - \xi({\bf t}, w)} \langle 0 | \hat{\gamma}_+({\bf t} - [z^{-1}] + [w^{-1}] )\hat{g} | 0 \rangle
       \over z-w},
       \eea
       where we have used eq.~(\ref{gamma+psi_conj}) and then extended  to arbitrary $n$,  either by induction or, equivalently, by the fermionic Wick theorem.
  \end{proof}
  
 Using Wick's theorem,  we have the following determinantal expression in terms of  single pair correlators
and  in terms of the $\tau$-function.
  \begin{corollary}
  \label{K2n_det}
  \be 
  \tilde{K}^g_{2n} ({\bf z}, {\bf w},{\bf t})
   = \det_{i,j=1}^n \left({ e^{\xi({\bf t,} z_i)-\xi({\bf t}, w_j)}   \tau_g({\bf t} -[{z_i}^{-1}] + [{w_j}^{-1}]) \over z_i - w_j} \right).
  \ee\label{correl_K_det}
    \end{corollary}
From the fermionic representation, it is also easy to see that the 2-point function at
 ${\bf t}= {\bf 0}$
may be expressed as a series
\begin{proposition}
\label{K_w_j_series} 
\be
\tilde{K}^g_2(z, w) = \begin{cases}  \sum_{j=1}^\infty w^g_j(w)   w^{g*}_{-j+1}(z) \quad \text{if} \ |z| > |w|, \\
- \sum_{j=1}^\infty w^g_{-j+1}(w)   w^{g*}_{j}(z) \quad \text{if} \ |z| <|w|.
\end{cases}
\label{pair_correl_expansion_w_g_k}
\ee
\end{proposition}
\begin{proof}
A fermionic proof is given in Appendix \ref{K2_expansion_adapted_basis}.  Alternatively, eq.~(\ref{pair_correl_expansion_w_g_k})  may be proved using the $n=1$ case of \eqref{correl_K_det} and the Schur function expansion \eqref{tau_schur_exp} of the $\tau$-function.
\end{proof}

We also have the following {\em Christoffel-Darboux} type formula.
\begin{proposition}
\label{CD_rep_prop}
\be
\tilde{K}^g_2(z,w) = { \sum_{i=0}^\infty \sum_{j=0}^\infty (\tilde{A}_{ij} w^{g}_{-j}(z) w_{-i}^{g*}(w))
- \tilde{Q}^+_{01}w^g_1(z) w^{g*}_1(w) \over z -w},
\label{pair_correl_expansion_w_g_alt}
\ee
where
\be
\tilde{A}_{ij} = \tilde{Q}^+_{j+1, -i}, 
\ee
and $\tilde{Q}^+_{ij}$ is the recursion matrix appearing in eq.~(\ref{tilde_Q_+}).
\end{proposition}
\begin{proof} Multiply the RHS of eq.~(\ref{pair_correl_expansion_w_g_k}) by $z-w$.
Then the result follows  from using the recursion relations  (\ref{zwk_recursion}) directly, 
together with the telescopic cancellation between the two types of terms coming from the $z$ and $w$ recursion relations.
Alternatively, it may be derived  directly from the VEV representation eq.~(\ref{K_2n_time}) by using the fermionic identity
\be
(z-w) \psi(w) \psi^*(z) = [J_1, \psi(w) \psi^*(z)]
\ee
and inserting sums over complete sets of intermediate states in both terms arising from the commutator.
 \end{proof}

\subsection{Multicurrent correlator for general $\hat{g}$}
 \label{F_g}
 In the following, it will be convenient to use the inverse variables $x_i := 1/z_i$ in expressing
 the correlators. 
 \begin{definition}
 We denote the positive frequency part of the current operator as
 \be
 J_+(x) := \sum_{i=1}^\infty x^i J_i,
 \ee
 and define the multicurrent correlator $\JJ_n(x_1, \dots, x_n)$ as
 \be
 \JJ_n(x_1, \dots, x_n) := \langle 0 | \left(\prod_{i=1}^n J_+(x_i)\right) \hat{g}| 0 \rangle.
 \label{multicurrent_correl}
 \ee
 \end{definition}

 We also introduce the correlators $W_n(x_1, \dots, x_n)$, defined  as multiple derivatives of the $\tau$-function
 \be
W_n(x_1, \dots, x_n) := \left.\left(\prod_{i=1}^n \nabla(x_i)\right) \tau_g({\bf t})\right|_{{\bf t} = {\bf 0}},
\label{W_n_def}
 \ee
 where the commuting parametric family of first order gradient operators $\nabla(x_i)$  in the ${\bf t}$-variables is defined by
 \be
 \nabla(x) := \sum_{i=1}^\infty x^{i-1}{\partial \over \partial t_i}.
 \ee
 For later purposes, we also introduce a ``connected'' version of these quantities, 
 defined by replacing $\tau_g$ by $\ln (\tau_g)$
 \be
\tilde{W}_n(x_1, \dots, x_n) :=\left.\left( \prod_{i=1}^n \nabla(x_i) \right) (\ln (\tau_g({\bf t}))\right|_{{\bf t} = {\bf 0}}.
 \ee
Assuming $\tau_g({\bf t})$ to be normalized such the $\tau_g({\bf 0}) = 1$, we have, in particular,
 \bea
 \tilde{W}_1(x_1) &\& = W_1(x_1 ) \cr
  \tilde{W}_2(x_1, x_2) &\& = W_2(x_1, x_2)  - W_1(x_1) W_2(x_2) \cr
   \tilde{W}_3(x_1, x_2, x_3) &\& = W_3(x_1, x_2, x_3 )  - W_1(x_1) W_2(x_2, x_3) - W_1(x_2) W_2(x_1, x_3)
   - W_1(x_3) W_2(x_1 x_2) \cr
   &\&{\hskip 10 pt} + 2 W_1(x_1) W_1(x_2) W_1(x_3)\cr
  \vdots  {\hskip 15 pt}&\& =    {\hskip 20 pt} \vdots 
  \label{tilde_W123}
 \eea
From the fermionic formula (\ref{tau_g_VEV}) for $\tau_g({\bf t})$  it follows that $W_n$
and the multicurrent correlator $\JJ_n$ are related by
 \begin{lemma}
 \be
 \label{W_n_JJ_n}
\JJ_n(x_1, \dots, x_n)  = \left( \prod_{i=1}^n x_i \right)W_n(x_1, \dots , x_n).
\ee
 \end{lemma}
 \begin{proof}
 This is immediate from applying (\ref{W_n_def}) to the fermionic formula (\ref{tau_g_VEV}) for $\tau_g({\bf t})$.
 \end{proof}

\section{Specialization to the case of hypergeometric $\tau$-functions}
\label{Hypersection}


\subsection{Convolution products}
An abelian group that is central to hypergeometric $\tau$-functions is obtained by extending the semigroup 
 $C:= \{C_\rho\}$ consisting of formal convolution products with functions or distributions on the circle $S^1$ admitting
 a Fourier series representation
\be
\rho(z) = \sum_{i= -\infty}^{\infty} \rho_i z^{-i-1}
\label{rho_fourier}
\ee 
 i.e.  the diagonal action on an element $w\in \HH$ with Fourier representation
\be 
w(z) := \sum_{i= -\infty}^{\infty} w_i z^{-i-1},
\ee
consisting of multiplication of their Fourier coefficients
\be
C_\rho(w)(z):= \sum_{i=-\infty}^{\infty} \rho_i w_i z^{-i-1},
\label{crho}
\ee
to include the diagonal action of commuting differential operators of the form
\be
\exp\left(\sum_i \alpha_i \DD^i\right),
\ee
where $\DD$
is the Euler operator (\ref{Eulerop}) (or simply the vector field of infinitesimal rotations
on the circle). (These may also formally be represented as a convolution with the series $\rho(z)$ in which
the coefficients in (\ref{rho_fourier}) are $\rho_i=\exp\left(\sum_j \alpha_j (-i-1)^j\right)$.)

In the $L^2(S^1)$ sense, the  convolution product may be viewed
as the integral
\be
C_\rho( w) (z) := {1\over 2\pi i} \oint_{S^1} \rho (\zeta) w\left({z\over  \zeta}\right) {d\zeta \over \zeta}
= {1\over 2\pi i} \oint_{S^1} \rho \left({z\over  \zeta}\right) w(\zeta) {d\zeta \over \zeta}.
\label{convol_rho_w}
\ee
Alternatively, in the formal Laurent series sense, the integral 
$ {1\over 2\pi i} \oint_{S^1}f(\zeta) {d\zeta \over \zeta}$ may be reinterpreted as a formal residue  at $z=0$.

  These are essential to the definition of KP or $2D$-Toda $\tau$-functions
of hypergeometric type \cite{OrSc1,OrSc2}, since they correspond to Grassmannian elements of the form 
\be
W^{(\rho, {\bf s})} := C_\rho \gamma_-({\bf s})(\HH_+) \in  \Gr_{\HH_+}(\HH)
\ee
for the zero virtual dimension component or, more generally
\be
W_N^{( \rho, {\bf s})}  = C_\rho \gamma_-({\bf s})(z^{-N}\HH_+),  \quad N\in \Zb
\ee
in the virtual dimension $N$ sector. 

In the following, the coefficients
\be
\rho_i =: e^{T_i}
\label{rho_i_T_i}
\ee
in  the distributional (or formal) series $\rho(z)$ will always be  understood as
determined  by  the weight generating function $G$ and the constants $(\beta, \gamma)$, 
as in   \cite{H2},  such that
\be
 r^{(G, \beta)}_i  = G(i \beta ) = {\rho_i \over \gamma \rho_{i-1}} .
\label{expT_ratio}
\ee
Since the $\tau$-function is only projectively defined, assuming $\rho_0 \neq 0$, we may choose it  as $\rho_0=1$
and express the coefficients $\rho_j$  as finite products of the  $ \gamma G(i\beta )$'s and their inverses  \cite{H2} 
\be
\rho_j= \gamma^{j}\prod_{i=1}^j G(i\beta), 
\quad \rho_{-j} = \gamma^{-j} \prod_{i=0}^{j-1} (G(-i \beta))^{-1}, \quad j=  1, 2, \dots.
\label{rho_j_gamma_G}
\ee
Thus
\bea
T_j(\beta, \gamma) &\&:= j \ln(\gamma) + \sum_{i=1}^j \ln (G(i\beta)), \quad j=1, 2, \dots, \quad T_0 :=0 \cr
T_{-j}(\beta, \gamma) &\&:= -j \ln(\gamma) - \sum_{i=0}^{j-1} \ln (G(-i\beta)), \quad j=1, 2, \dots
\label{Tj_G_def}
\eea
Although vanishing values  $\rho_{j}=0$   for some negative coefficient $j=-n-1$ are allowed,
 the  relations {\eqref{rho_j_gamma_G} then imply that
\be
\rho_j = 0 \quad  \forall\  i \le -n-1
\ee
and hence all the coefficients  $\gamma^{|\lambda| }r_\lambda^{(G,\beta)}$ appearing in the
$\tau$-functions ({\ref{tau_G_H}), ({\ref{tau_tilde_G_H}) vanish for 
partitions of length $\ell(\lambda) >n$. 


\subsection{Adapted basis, $\tau$-function and Baker function for the hypergeometric case}
\label{hypergeometric_specialization}

Since $\gamma_-({\bf s})$ is the generating function for the complete symmetric functions,
\be
\gamma_-({\bf s})= \sum_{i=0}^\infty h_i({\bf s})z^{-i},
\ee
specializing to the case 
\be
g:= C_\rho   \gamma_-(\beta^{-1}{\bf s}),
\label{hypergeom_g}
\ee
we have
\be
g_{ij} = \rho_i h_{i-j} (\beta^{-1} {\bf s}), \quad g^{-1}_{ij} = \rho_j^{-1} h_{i-j}(-\beta^{-1}{\bf s}).
\label{hypergeom_g_matrix_els}
\ee

The corresponding adapted bases $\{w^{(G, \beta, \gamma, {\bf s})}_k(z)\}_{k\in \Zb}$ defined by applying the element  $C_\rho  \gamma_-(\beta^{-1}{\bf s})$  and its dual to the basis elements $\{e_{m} := z^{-m-1}\}_{m\in \Zb}$ are
\be
w^{(G, \beta, \gamma,  {\bf s})} _k (z) := C_{\rho}  (\gamma_-(\beta^{-1}{\bf s}) z^{k-1}) = \sum_{j=-\infty}^{k-1} h_{k-j-1}({\beta^{-1}\bf s})  \rho_{ -j-1} z^j, 
\label{wrhok}
\ee
where the coefficients $ \rho_j$ are defined  in eq.\,\eqref{rho_j_gamma_G}. The dual basis, under the Hirota bilinear pairing (\ref{hirota_bilinear}), is defined by
\be
w_k^{(G, \beta, \gamma, {\bf s})*}(z) := \iota  C^{-1}_{\rho} (\gamma_+(-\beta^{-1}{\bf s})  \iota( z^{k-1})) 
=  \sum_{j=-\infty}^{k-1} h_{k-j-1}(-\beta^{-1}{\bf s})  \rho^{-1}_{ j} z^j, 
\label{wrhosk}
\ee
where $\iota$ is the involution (\ref{involution}). 	Simplifying the notation for the adapted bases in this case to 
\be
w^{(G, \beta, \gamma,  {\bf s})} _k (z) =: w_k(z), \quad w^{(G, \beta, \gamma,  {\bf s})*} _k (z) =: w^*_k(z)
\ee	
these have the fermionic representation (cf. \eqref{wkg_fermionic})  
\bea
w_k(z) &\& :=\begin{cases} \langle 0 | \psi_{-k} \hat{\gamma}^{-1}_-({\bf s})\ \hat{C}^{-1}_\rho \psi^*(z) \hat{C}_\rho \hat{\gamma}_-(\beta^{-1}{\bf s})| 0\rangle \quad  \text{if} \ k\ge 1, \cr
\langle 0 |\psi^*(z) \hat{C}_\rho \hat{\gamma}_-(\beta^{-1}{\bf s}) \psi_{-k} | 0\rangle  \quad \text{if} \ k \le 0, 
\end{cases} 
\label{wk_rho_s}
\eea
which coincides with the basis defined in \eqref{wrhok}, 
and the dual basis, defined by \eqref{wkg*_fermionic},  becomes
\bea
w_k^{*}(z)  &\&:=  \begin{cases} \langle 0 | \psi^*_{k-1} \hat{\gamma}_-({-\beta^{-1}\bf s})\hat{C}^{-1}_\rho \psi(z) \hat{C}_\rho \hat{\gamma}_-(\beta^{-1}{\bf s})|0 \rangle  \quad  \text{if} \ k\ge 1 \cr
\langle 0 |\psi(z) \hat{C}_\rho \hat{\gamma}_-(\beta^{-1}{\bf s}) \psi^*_{k-1} |0 \rangle \quad  \text{if} \ k\le 0, \\ 
\end{cases}
\label{wk_rho_s*}
\eea
which coincides with  \eqref{wrhosk}.

For $k\in \Nb^+$ they define corresponding elements of the Grassmanian:
\bea
W^{(G, \beta, \gamma, {\bf s})}=\span \{w_k(z)\}_{k\in \Nb^+}\\
W^{(G, \beta, \gamma, {\bf s})\perp}=\span \{w^{*}_k(z)\}_{k\in \Nb^+}
\eea

The fermionic representation of the group element  associated to the weight generating function $\gamma G$ defined in \eqref{rho_j_gamma_G} is of the form
\be
\hat{g}:= \hat{C}_\rho  \hat{\gamma}_-(\beta^{-1}{\bf s}) 
 \label{g_C_hat_rho_gamma-}
\ee
where
\be
\hat{C}_\rho = e^{\hat{A}}
\ee
with the algebra element $\hat{A}$ is
\be
\hat{A}:=\sum_{i \in \Zb} T_i \no{\psi_i \psi^*_i}
\label{A_hat_def}
\ee

We therefore have the following fermionic representation of the  2D Toda 
$\tau$-function $ \tau^{(G, \beta, \gamma)} ({\bf t}, {\bf s})$ at $N=0$ defined by \eqref{tau_G_H} 
\be
\tau^{(G, \beta, \gamma)}({\bf t}, {\bf s}) = \langle 0 | \hat{\gamma}_+({\bf t}) \hat{C}_\rho \hat{\gamma}_-({\bf s})| 0 \rangle.
\ee
.\begin{remark}
Such 2D Toda $\tau$-functions, referred to in \cite{OrSc1, OrSc2} as {\em $\tau$-functions of hypergeometric type},
 were considered, for other purposes, in \cite{KMMM}, but without identification of these as Hurwitz generating functions.
 Fermionic representations of particular families of the Hurwitz numbers were considered in  \cite{OrSc1,OrSc2,  Ok, GH2, H1, H2,  HO2, Ta1, AMMN, ALS, NOr, OP,MSS}
\end{remark}

Viewed as a parametric family of KP $\tau$-functions, after the substitution ${\bf s} \ra \beta^{-1}{\bf s}$, the corresponding Baker function and its dual have the fermionic representation
\bea
\Psi^-_{(G, \beta, \gamma)}(z, {\bf t}, \beta^{-1}{\bf s}) &\&={ \langle 0 | \psi^*_0 \hat{\gamma}_+({\bf t}) \psi(z) \hat{C}_\rho \gamma_-(\beta^{-1}{\bf s})| 0 \rangle
\over \tau^{(G, \beta, \gamma)}({\bf t}, {\beta^{-1}\bf s}) },  \\
\label{Baker_fermionic_hypergeometric}
\Psi^+_{(G, \beta, \gamma)}(z, {\bf t}, \beta^{-1}{\bf s}) &\&= {\langle 0 | \psi_{-1}\hat{\gamma}_+({\bf t} )\psi^*(z) 
\hat{C}_\rho \gamma_-(\beta^{-1}{\bf s}) | 0\rangle \over  \tau^{(G, \beta, \gamma)}({\bf t}, {\beta^{-1}\bf s})  }. 
\label{dual_Baker_fermionic_hypergeometric}
\eea

As particular cases of the expansions \eqref{Psi_w_k_g_exp},  \eqref{Psi*_w_k_g_exp}, the Baker function and its dual for this class of hypergeometric $\tau$-functions can be expressed relative to the dual bases $\{ w_k^{*}(z)\}$,
$\{ w_k(z)\}$ as
\bea
\Psi^-_{(G, \beta, \gamma)}(z, {\bf t},\beta^{-1} {\bf s})  &=&\sum_{k=1}^\infty w_k^{*}(z) \alpha_k^{(G, \beta, \gamma)}({\bf t},{\bf s}), \\
\Psi^+_{(G, \beta, \gamma)}(z, {\bf t},\beta^{-1} {\bf s}) &=&  \sum_{k=1}^\infty w_k(z) \beta_k^{(G, \beta, \gamma)}({\bf t},{\bf s}),
\eea
where
\bea
\alpha_k^{(G, \beta, \gamma)}({\bf t},{\bf s})&:=& {\langle 0 | \psi^*_0 \hat{\gamma}_+({\bf t}) \hat{C}_\rho \hat{\gamma}_-(\beta^{-1}{\bf s}) \psi_{k-1}| 0 \rangle  \over \tau_g({\bf t},\beta^{-1}{\bf s})}, \\
\beta_k^{(G, \beta, \gamma)}({\bf t},{\bf s})&:=& 
{ \langle 0 | \psi_{-1} \hat{\gamma}_+({\bf t}) \hat{C}_\rho\hat{\gamma}_-(\beta^{-1}{\bf s}) \psi^* _{-k}| 0\rangle   \over \tau_g({\bf t},\beta^{-1}{\bf s})}.
\eea


\subsection{An alternative realization of convolution actions }
\label{alternative_convolution}

It is easy to see that for (\ref{rho_j_gamma_G}) the convolution operator can formally be equivalently realized
 by the exponential of a differential operator, namely
\be
C_\rho= e^{T(-{\mathcal D}-1)}
\ee
where the $T(x)$ is viewed as a formal Taylor series satisfying
\be
\label{TGrel}
e^{T(x)-T(x-1)}=\gamma G(\beta x),
\ee
and $T(0)=0$ with
\be
T(i)=T_i \quad \text{for}\ i\in \Zb. 
\label{T_integer}
\ee 

This series, though nonunique, may be given an explicit expression as
\be\label{Tfuncd}
T(x)=\sum_{k=0}^\infty (-1)^{k+1} \beta^k A_k\, p_k(x),
\ee
where 
\bea
A_0=-\log \gamma,\\
A_k=\frac{1}{k} \sum_{j=1}^\infty c_j^k, \,\,\,\,k>0,
\eea
the polynomials $\{p_k(x)\in x\,\mathbb C[x]\}$ are determined by 
\be
p_k(x)-p_k(x-1)=x^k
\ee
and are closely related to the Bernoulli polynomials. In particular, 
\bea
p_0(x)&=&x,\\
p_1(x)&=&\frac{1}{2}x \left( x+1 \right),\\
p_2(x)&=&\frac{1}{6}x \left( x+1 \right)  \left( 2\,x+1 \right),\\
p_3(x)&=&\frac{1}{4}{x}^{2} \left( x+1 \right) ^{2},\\
p_4(x)&=&{\frac {1}{30}}\,x \left( x+1 \right)  \left( 2\,x+1 \right)  \left( 
3\,{x}^{2}+3\,x-1 \right), \\
&\vdots  &
\eea
Their generating function is 
\be
\frac{1}{1-e^{-a}}(e^{ax}-1). =\sum_{k=0}^\infty \frac{p_k(x)a^k}{k!}.
\ee
\begin{remark}
There are, of course, an infinity of other series $T(x)$ satisfying eqs.~(\ref{TGrel}) and (\ref{T_integer}),
obtained, e.g. by adding any periodic function $P(x)$ of unit period satisfying $P(0)=0$.
\end{remark}

The fermionic operator  $\hat{A}$ appearing in eqs.~(\ref{A_hat_def}), may equivalently be expressed as
\be
\hat{A} = \mbox{res}_{z=0}\left(\no{\psi(z) T\left(-\mathcal{D}-1\right) \psistar(z)}\right),
\label{Ahat_T_def}
\ee
which satisfies the commutation relations
\bea\label{com}
\left[\hat{A},\psi(z)\right]&=& T\left({\mathcal D}\right) \psi(z),\\
\left[\hat{A},\psistar(z)\right]&=&-T\left(-{\mathcal D}-1\right) \psistar(z).
\eea

The adapted basis elements can equivalently be expressed as 
\bea
w _k (z)= e^{T(-{\mathcal D}-1)} \gamma_-(\beta^{-1}{\bf s}) z^{k-1},\\
w _k^* (z)= e^{-T({\mathcal D})} \gamma_-(-\beta^{-1}{\bf s}) z^{k-1}.
\label{Dif_op_w}
\eea

The linearly independent elements $\{w_{N+ k }\}_{k\in \Nb^+}$  span an infinite nested sequence
of subspaces 
\be
W_N^{( G, \beta, \gamma, {\bf s})} := \span\{w_{N+ k}\}_{k\in \Nb^+}
\ee 
of virtual dimension $N$ 
\be
\cdots \ss W_{N-1}^{(G, \beta, \gamma, {\bf s})} \ss W_N^{( G, \beta, \gamma, {\bf s})} \ss W_{N+1}^{(G, \beta, \gamma, {\bf s})} \ss \cdots.
\ee
This defines an infinite complete  flag which, for fixed parameters ${\bf s}$ and convolution
elements $\rho$,  determines a chain of $\tau$-functions of the mKP hierarchy,  defined as  
determinants of the corresponding orthogonal projection maps
 \bea
\tau^{(G,\beta, \gamma)} (N, {\bf t}, \beta^{-1} {\bf s}) &\&:= \tau_{W_N^{( G, \beta,\gamma,{\bf s})}} (N, {\bf t})   \\
&\& \cr
&\&=    \sum_{\lambda} \gamma^{|\lambda|} r^{(G, \beta)}_{(\lambda, N)} s_\lambda ({\bf t}) s_\lambda (\beta^{-1} {\bf s}), 
 \label{tau_G_H_N}
\eea
where the content product coefficient  $r^{(G, \beta)}_{\lambda}$, in eq.~\eqref{tau_G_H}  is
 replaced by the shifted version
 \be
r_{(\lambda, N)}^{(G, \beta)} \deq   \prod_{(i,j)\in \lambda} G(\beta(N+ j-i)).
\label{r_lambda_G_N} 
\ee
Alternatively, viewing the parameters ${\bf s} = (s_1, s_2 \dots )$ as a further infinite set of flow
variables, this defines a hypergeometric  $\tau$-function of the  $2D$ Toda hierarchy. Viewing all parameters  $\{N,  {\bf t}, {\bf s}\}$ on the same footing, as discrete or continuous flow parameters, we obtain a lattice of $\tau$-functions of the $2D$-Toda hierarchy.


\subsection{Recursion relations in the hypergeometric case}
\label{recursion_relations}

We also introduce an alternative notation for these bases which is used in refs.~\cite{ACEH1, ACEH2}
\be
\Psi^+_k(x) := \gamma w_{1-k}(z = 1/x), \quad \Psi^-_k(x) := w^*_{1-k}(z=1/x), 
\label{Psi_pm_k_ser}
\ee
where
\be
x = {1\over z}.
\ee

The basis elements $\{\Psi^\pm_k(x)\}$ defined in (\ref{Psi_pm_k_ser})  are viewed as elements of the subspaces 
 \be
 W^\pm_{(G, \beta, \gamma, {\bf s})} := \span\{ \Psi^\pm_{-k} (x)\}_{k\in \Nb} \ss
\mathbb K[x, x^{-1}, {\bf s},\beta,\beta^{-1}]((\gamma))
 \ee
 in the formal series sense. They may be understood as elements of formal power series 
 analogs of the Sato-Segal-Wilson Grassmannian, consisting of subspaces of the Hilbert space spanned
 by powers $\{z^k\}_{k\in \Zb}$ of  the spectral parameter $z = 1/x$, commensurable with the subspace 
 spanned by the monomials  $\{z^k\}_{k=0, 1, \dots}$.

 The Baker function $\Psi^-_{(G, \beta, \gamma)}(z= 1/x, {\bf t }, \beta^{-1} {\bf s})$
 and its dual $\Psi^+_{(G, \beta, \gamma)}(z= 1/x, {\bf t }, \beta^{-1} {\bf s})$,
are given by the Sato formula
 \be
 \Psi_{(G, \beta, \gamma)}^\mp(z, {\bf t}, {\bf s}) = e^{\pm \sum_{i=1}^\infty t_i z^i }{\tau^{(G, \beta, \gamma)}({\bf t} \mp [x] , {\bf s}) \over  \tau^{(G, \beta, \gamma)}({\bf t}, {\bf s}) }
 = e^{\pm \sum_{i=1}^\infty t_i z^i } \left(1 + \OO({1/ z})\right),
 \ee
 These may be viewed as  elements of the subspaces $W^\pm_{(G, \beta, \gamma, {\bf s})}$
 for all values of the KP flow parameters ${\bf t}$.

Specializing eqs.~(\ref{zwk_recursion}, \ref{zwk_recursion_dual}) to this case, the basis elements 
$\{\Psi^+_i(x), \Psi^-_i(x)\}$ satisfy the recursion relations
\begin{proposition}
\label{Psi_recurs_prop}
\bea
{1\over \gamma x}\Psi^+_i &\&= \sum_{j=i-1}^\infty Q^+_{ij} \Psi^+_j,  
\label{mult_rr_Psi+}\\
{1\over \gamma x}\Psi^-_i &\&= \sum_{j=i-1}^\infty Q^-_{ij} \Psi^-_j ,
\label{mult_rr_Psi-}
\eea
where the (nearly) upper triangular matrices  $Q^\pm$ are the negatives of the transposes of $\tilde{Q}^\mp$
\be
Q^+_{kj} = \gamma^{-1}\tilde{Q}^-_{jk}, \quad Q^-_{kj} =  \gamma^{-1} \tilde{Q}^+_{jk}.
\label{mult_rec_matrices_general}
\ee
In this case we can express them explicitly as
\bea
Q^+_{ij} &\&:=  \sum_{k=i-1}^{j} r^{G}_k(\beta) h_{k-i+1}(\beta^{-1}{\bf s})  h_{j-k}(-\beta^{-1}{\bf s}) , \quad j\ge i-1\\
Q^-_{ij} &\& :=  \sum_{k=i-1}^{j}  r^{G}_{-k}(\beta) h_{j-k}(\beta^{-1}{\bf s}) h_{k -i+1}(-\beta^{-1}{\bf s}),  \quad j\ge i-1.
\label{Q_pm_hypergeometric}
\eea
\end{proposition}

Applying the Euler operator $\DD$ to the basis elements $\{w_k, w_k^{*}\}$, we obtain the
following differential linear relations
\begin{proposition}
\label{wk_dif_rrs}
\bea
\DD w_k(z) &\&= \sum \tilde{P}^+_{kj} w_j(z)  \\
\DD w_k^{*}(z) &\&= \sum \tilde{P}^-_{kj} w_j^{*}(z),
\label{D_recursions_w_k}
\eea
where the elements of the lower triangular matrices $\tilde{P}^\pm$ are
\be
\tilde{P}^\pm_{ij} = 
\begin{cases} (i-1)\delta_{ij} \mp \beta^{-1}(i-j) s_{i-j}, \quad i\ge j \cr
    0, \quad i<j.
\end{cases}
\label{tilde_Pij}
\ee

\end{proposition}
\begin{proof}
To see this, we may either use the general formulae (\ref{tilde_P_pm}) of Section \ref{deriv_rec_rels},
specialized to the case (\ref{hypergeom_g_matrix_els}) or, more simply, proceed directly. Namely, 
from (\ref{Dif_op_w}) we have
\bea
\DD w_k(z)&=&e^{T(-{\mathcal D}-1)}  \DD \gamma_-(\beta^{-1}{\bf s}) z^{k-1}
\label{DDwkT-rels}\\
&=&e^{T(-{\mathcal D}-1)} (k-1-\beta^{-1} S(z^{-1}))  \gamma_-(\beta^{-1}{\bf s}) z^{k-1}
\cr
&=&(k-1) w_k(z)-\beta^{-1} \sum_{j=1}^\infty j s_j w_{k-j}(z),
\label{wk_diff_eq}
\eea
where
\be
S(z):=\sum_{k=1}^\infty k s_k z^{k}
\label{S_def}
\ee
and similarly,
\be
\DD w_k^*(z)=(k-1) w^*_k(z)+\beta^{-1} \sum_{j=1}^\infty j s_j w^*_{k-j}(z).
\label{wk*_diff_eq}
\ee
This linear action is represented by the matrices $\tilde{P}^\pm$ defined in (\ref{tilde_Pij}). 
\end{proof}

In terms of the notation $\{\Psi_k^+(x)\}$, $\{\Psi_k^-(x)\}$, the differential linear relations are
\bea
\beta D \Psi_k^+(x) &\&= \sum_{j=-\infty}^{k+1} P^+_{kj} \Psi_j^+(x) \\
\beta D\Psi_k^-(x) &\&= \sum_{j=-\infty}^{k+1} P^-_{kj} \Psi_j^-(x).
\label{deriv_rr_Psi_general}
\eea
where
$D$ is the Euler operator in the $x$-variable
\be
D:= x {d \over dx} = -\DD.
\ee
and the upper triangular matrices $P^\pm$ are the negatives of the transposes of $\tilde{P}^\mp$
\bea
P^\pm_{ij} &\&= -\beta  \tilde{P}^\pm_{ji},   \cr
&\&= \begin{cases} \beta (1-i)\delta_{ij} \mp (i-j) s_{j-i}, \quad i\le j \cr
   0, \quad i>j. 
 \end{cases}
\label{deriv_rec_matrices_general} 
\eea


\subsection{The spectral curve and Kac-Schwarz operators }
\label{quantum_spectral_KS}

\subsubsection{The quantum and classical spectral curve}
\label{quantum_spectral}

Now define the differential operator $R$ acting either on $\HH$, or formally, 
on $\Cb[[z, z^{-1}]]$, as:
\be 
R  := \frac{\gamma}{z}  G(-\beta \DD),
\ee
 which acts on the monomial $z^i$:
\be
R (z^i) = \frac{\gamma}{z} G(-i\beta) z^i = {\rho_{-i} \over \rho_{-i-1}  }z^{i-1}.
\ee
Applying $ R$ to the series for $\rho(z)$, we see this is an eigenvector with eigenvalue $1$ 
\be
R\, \rho (z)  =  \rho(z).
\label{Q_G_def}
\ee
Applying it to $w_k$ gives
\be
R\,  w_k= w_{k-1}
\label{Rwk}
\ee
or iterating,
\be
(R)^j w_k= w_{k -j}.
\label{Rwkj}
\ee

Similarly, defining
\be
R^*:= \frac{\gamma}{z}  G(\beta \DD),
\ee
we have
\be
R^*\,  w^*_k= w^*_{k-1}
\label{Rwk*}
\ee
or iterating,
\be
(R^*)^j w^*_k= w^*_{k -j}.
\label{Rwkj*}
\ee
It is clear that the operators $R$ and $R^*$ are invertible.

Equivalently, in terms of the $x =1/z$ variables, we have the operators
\be
R_+ := \gamma x G(\beta D)  = R, \quad R_- = \gamma x G(-\beta D) = R^*
\ee
and eqs.(\ref{Rwkj}), (\ref{Rwkj*}) are equivalent to
\bea
(R_+)^j \Psi^+_k&\&= \Psi^+_{k+j} 
\label{R+Psi_k}\\
(R_-)^j \Psi^-_k&\&= \Psi^-_{k+j}
\label{R-Psi_k} 
\eea

Using (\ref{Rwkj}) and (\ref{Rwkj*}), we may  re-express eqs.~(\ref{wk_diff_eq}), (\ref{wk*_diff_eq}) more compactly as
\begin{proposition}
\label{wk_diff_eqs}
\bea
\left(\beta \DD +S(R)\right) w_k = (k-1) \beta w_k, 
\label{wk_spectral_eq}\\
\left(\beta \DD - S(R^*)\right) w^*_k = (k-1) \beta w^*_k.
\label{wk*_spectral_eq}
\eea
or, equivalently,
\bea
\left[\beta D - S(R_+)\right] \Psi_k^+(x) &\&=  k \Psi^+_k(x), 
\label{Psi+_k_spectral_eq}
\\
\left[\beta D + S(R_-)\right] \Psi_k^-(x) &\&= k \Psi^-_k(x).
\label{Psi-_k_spectral_eq}
\eea
\end{proposition}

\begin{remark}[{\bf Quantum and classical spectral curves}]
Since $\Psi^-_0(x) = w^*_1(z)$, $\Psi^+_0(x) = \gamma w_1(z)$ are the ${\bf t}= {\bf 0}$ evaluation
of the Baker function and its dual, respectively, the $k=0$ case of eqs.~(\ref{Psi+_k_spectral_eq}) and (\ref{Psi-_k_spectral_eq})
\bea
\left[\beta D - S(\gamma x G(+\beta D))\right] \Psi_0^+(x) &\&= 0, \\
\left[\beta D + S(\gamma x G(-\beta D))\right] \Psi_0^-(x) &\&= 0 ,
\label{quantum_spec_curve}
\eea
is satisfied by the Baker function and its dual at ${\bf t}= {\bf 0}$.

These are the {\bf quantum spectral curve} appearing in refs.~\cite{ACEH1, ACEH2}, and its dual.
Taking the classical limit amounts to replacing $\beta D$ by $xy$, where $y$ is the classical variable canonically conjugate to
$x$. Thus, the classical spectral curve is
\be
xy = S(\gamma x G(xy)),
\ee
and the dual is the same.
\end{remark}

\subsubsection{Kac-Schwarz operators}
\label{quantum_KS}

 \begin{definition}
If for a given $g$, an operator $a$ acting on ${\bf C}[[z,z^{-1}]]$ (e.g. a differential operator in $z$) stabilizes the element  $W_N^g = g(\HH^N_+)$ of the Sato Grassmannian 
(resp. the dual Grassmannian element $W_N^{g\perp}$)
\be
a W_N^{ g} \ss W_N^{ g}
\ee
or
\be
a W_N^{ g\perp} \ss W_N^{ g\perp},
\ee
we call it a {\em Kac-Schwarz operator} (resp. {\em dual Kac-Schwarz operator}) for the corresponding KP $\tau$-function  \cite{KS}. 
\end{definition}
\begin{remark}
Particular examples of Kac-Schwarz operators for weighted Hurwitz numbers were considered in \cite{AEG, ALS}.
\end{remark}

\begin{definition}
Define the following operators:
\bea
a&:=&e^{T(-1-\DD)} \gamma_-(\beta^{-1}{\bf s}) z \gamma_-(-\beta^{-1}{\bf s}) e^{-T(-1-\DD)}= e^{T(-1-\DD)} z e^{-T(-1-\DD)}\\
&=&z\frac{1}{\gamma G(\beta(-1-\DD))}=\frac{1}{R},\\
a^*&:=&e^{-T(\DD)} \gamma_-(-\beta^{-1}{\bf s}) z \gamma_-(\beta^{-1}{\bf s}) e^{T(\DD)}=e^{-T(\DD)}  z  e^{T(\DD)}\\
&=&z\frac{1}{\gamma G(\beta(1+\DD))}=\frac{1}{R^*},
\eea
\end{definition}
From (\ref{Dif_op_w}) we see that
\bea
 a\, w_{k} &=& w_{k+1}\\
a^*\, w_{k}^*  &=&w_{k+1}^* ,
\eea
so these are Kac-Schwarz (resp. dual Kac-Schwarz) operators, since they stabilize $W_N^{( G, \beta, \gamma,{\bf s})}$ (and $W_N^{( G, \beta,\gamma, {\bf s})\perp}$) for all $N$.

Similarly, define
\begin{definition}
\bea
b&:=&e^{T(1-\DD)} \gamma_-(\beta^{-1}{\bf s}) \DD \gamma_-(-\beta^{-1}{\bf s}) e^{-T(1-\DD)}= e^{T(1-\DD)} \left(\DD-\beta^{-1}S(z^{-1})\right) e^{-T(1-\DD)}\\
&=&\DD-\beta^{-1}S(R)\\
b^*&:=&e^{-T(\DD)} \gamma_-(-\beta^{-1}{\bf s}) \DD \gamma_-(\beta^{-1}{\bf s}) e^{T(\DD)}=e^{-T(\DD)}  \left(\DD+\beta^{-1}S(z^{-1})\right)  e^{T(\DD)}\\
&=& \DD+\beta^{-1}S({R^*}).
\eea
\end{definition}
Then (\ref{wk_spectral_eq}) and (\ref{wk*_spectral_eq}) are equivalent to
\bea
b\, w_k&=&  (k-1)w_k,\\
b^*\, w_k^*&=&   (k-1)w_k^*,
\eea
so these too are Kac-Schwarz operators for $W_N^{( G, \beta,\gamma {\bf s})}$ (and $W_N^{( G, \beta,\gamma {\bf s})\perp}$) 

Moreover, for the operators
\bea
c&:=&a^{-1}\,b=R \DD-\beta^{-1}R S(R) \\
c^{*}&:=&a^{*-1}\,b^*= R^* \DD + \beta^{-1}R^* S({R^*}),
\eea
we have
\bea
c\, w_k&=&  (k-1)w_{k-1}\\
c^*\, w_k^*&=&   (k-1)w_{k-1}^*.
\eea
Thus, these also are Kac-Schwarz (and dual Kac-Schwarz) operators for $W^{(G,\beta,\gamma,{\bf s})}$ and $W^{(G,\beta,\gamma,{\bf s})\perp}$. They satisfy the canonical commutation relations
\be
[c,a]=[c^*,a^*]=1
\ee
and completely specify the point of the Grassmanian.

More generally, for any $N\in \Zb$, the operators
\bea
c_N&:=&c -\frac{N}{a}, \\
c^{*}_N&:=&c^* -\frac{N}{a^*},
\eea
give
\bea
c_N\, w_k&=&  (k-1-N)w_{k-1},\\
c_N^*\, w_k^*&=&   (k-1-N)w_{k-1}^*,
\eea
satisfy the canonical commutation relations
\be
[c_N,a]=[c_N^*,a^*]=1
\ee
and stabilize  $W_{-N}^{( G, \beta,\gamma, {\bf s})}$ (and $W_N^{( G, \beta,\gamma, {\bf s})\perp}$).

\subsection{The pair correlation function in the hypergeometric case}
\label{hypergeom_K}

Specializing  the Christoffel-Darboux expresssion \eqref{pair_correl_expansion_w_g_alt}
for the pair correlation function $\tilde{K}_2^g(z,w)$  to the case (\ref{g_C_hat_rho_gamma-}), it may be expressed as follows
\begin{proposition}
\label{CD_hypergeom_prop}
\be
\tilde{K}^{(G, \beta, \gamma, {\bf s})}_2(z, w)  = { \sum _{i=0}^\infty\sum_{j=0}^\infty \tilde{A}^{(G, \beta, {\bf s})}_{ij} w_{-j}(z) w^{*}_{-i}(w)
 -  w_1(z) w_1^{*}(w)  \over z-w}
\label{pair_correl_expansion_w_rho_s_k} 
\ee
where
\be
 \tilde{A}^{(G, \beta, {\bf s})}_{ij} := \sum_{k=-j-1}^{i+1} r^{G}_k (\beta) h_{i-k+1} (-{\beta^{-1}\bf s}) h_{j+k+1}(\beta^{-1}{\bf s}).
 \label{tilde_A_ij_hypergeom_def}
\ee
\end{proposition}
Here we have used the fact that, for $g= C_\rho \gamma_-(\beta^{-1}{\bf s})$,$g_{00} g^{-1}_{-1, -1} = \gamma$.

Alternatively, in terms of the notation $\{\Psi^{\pm}_j\}_{j\in \Nb}$
\be
\tilde{K}^{(G, \beta, \gamma, {\bf s})}_2\left({1\over x},  {1\over x'}\right) =  {x x' \gamma^{-1}\left( \sum _{i,j=0}^\infty \tilde{A}^{(G, \beta, {\bf s})}_{ij} \Psi^+_{j+1}(x) \Psi^-_{i+1}(x') 
-  \Psi^+_0(x)\Psi_0^-(x')\right)\over x - x'}.
\ee

Equivalently, for 
\be
K(x, x') := -\gamma (x x')^{-1}\tilde{K}^{(G, \beta, \gamma, {\bf s})}_2\left({1\over x},  {1\over x'}\right),
\ee 
we have
\be
K(x, x') = {\sum_{i,j=0}^{\infty} A_{ij} \Psi^+_i(x) \Psi^- _j(x') \over x - x'},
\label{K_Psi+_Psi-denom}
\ee
where
\bea
A_{ij} &\&= - \tilde{A}^{(G, \beta, {\bf s})}_{j-1, i-1} =  - \sum_{k=-i}^{j} r^{(G, \beta)}_k h_{j-k} (-{\beta^{-1}\bf s}) h_{i+k}(\beta^{-1}{\bf s}).
\quad i,j = 1,2, \dots, \label{Acoefficients}
\\
A_{00} &\&=1, \quad A_{0j} = A_{i 0} = 0.
\nonumber
\eea

\subsection{Finiteness and generating function  for Christoffel-Darboux matrix}
\label{finite_rank}

In this section we prove that, for polynomial $G(z)$ and $S(z)$, the numerator of the  correlator (\ref{pair_correl_expansion_w_rho_s_k}), viewed as an integral
operator kernel into $W$,  is of finite rank, and that in fact the expression may be written as a finite sum.

Define the following power series in two variables $(r,t)$
\be\label{Aseriesdef}
A(r,t):=\left(r\,G\left(S(t)-\beta t\frac{d}{d t} \right)  -t  G\left(S(r) +\beta r\frac{d}{d r }\right)\right) \frac{1} {r-t}
\ee
\begin{proposition}
\label{generating_function_A}
The generating function for the Christoffel-Darboux coefficients $A_{ij}$ defined by (\ref{Acoefficients}) is given by
\be
A(r,t)=\sum_{i=0}^{\infty} \sum_{j=0}^{\infty} A_{ij}r^i t^j.
\ee
In particular, if $s_i =0, \ \forall\  i > L$ and $g_i=0, \ \forall\  i > M$, then
\be\label{Azero}
A_{ij}=0\,\, \forall\ \,\, i+j > LM.
\ee
\end{proposition}

To prove this, we begin with the following easily proved generalization of the orthogonality relations satisfied
by the complete symmetric functions, valid in the case when only the first $L$ variables $(s_1, \dots, s_L, 0, 0, \cdots)$
are nonzero.
\begin{lemma}
\label{general_h_n_orthogon}
\bea
e^{\xi({\bf s}, x)}D^k e^{-\xi({\bf s}, x)}&=& \sum_{N=0}^\infty x^N\sum_{n=0}^N n^k  h_n(-{\bf s})  h_{N-n}({\bf s})\\
e^{-\xi({\bf s}, x)}D^k e^{\xi({\bf s},x)}&=& \sum_{N=0}^\infty x^N \sum_{n=0}^N n^k  h_n({\bf s})  h_{N-n}(-{\bf s})
\label{lemma_orthog}
\eea
Moreover, if $s_i =0, \ \forall\  i > L$, the following orthogonality relations are identically satisfied
\be
\sum_{n=1}^N n^k h_n(-{\bf s}) h_{N-n}({\bf s}) =0, \quad \forall N > kL.
\label{h_nk_orthogonality}
\ee
\end{lemma}
\begin{proof}
For each $N\ge 0$, define the following generating series
\bea
S(x,y) &\& :=  e^{\xi(s,y)-\xi({\bf s},x)}\\
&\&=:\sum_{N=0}^\infty S_N(x,y)
\eea
where the $S_N(x,y)$  are  polynomials of degree $N$ in the variables $(x,y)$
\be
S_N(x,y) :=   \sum_{n=0}^N x^n y^{N-n} h_n(-{\bf s})  h_{N-n}({\bf s}).
\ee
For all  $k\in \Nb$, we also define two weighted versions of the polynomials $S_{N}(x,y)$
\be
S_{N,k}(x,y) :=   \sum_{n=0}^N n^k x^n y^{N-n} h_n(-{\bf s})  h_{N-n}({\bf s})
\ee
Denoting the Euler operator in $x$ 
\be
D :=x {d\over dx}\\
\ee
and applying its $k$th power to $S(x,y)$, we obtain
\be
D^k S(x,y) = P(x) S(x,y) =   \sum_{N=0}^\infty S_{N, k}(x,y),
\ee
where $P(x)$  is the series
\be
P(x)=e^{\xi({\bf s},x)}D^k e^{-\xi({\bf s},x)}.
\label{Pseries}
\ee
For $x=y$ we have
\be
P(x)=  \sum_{N=0}^\infty x^N\sum_{n=0}^N n^k  h_n(-{\bf s})  h_{N-n}({\bf s})
\label{PPtilde}
\ee
which proves the first equation in  \ref{lemma_orthog}. The second equation follows from the substitution ${\bf s}\rightarrow -{\bf s}$.

If $s_i =0, \ \forall\  i > L$, then  $P(x)$ and is a polynomial of degree $kL$. In this case (\ref{PPtilde})  implies (\ref{h_nk_orthogonality}).
\end{proof}

We now turn to the proof of Proposition \ref{generating_function_A}.  
\begin{proof}
It is obvious that it is  sufficient to prove it for $\beta=1$. For (\ref{Aseriesdef}) we have
\bea
A(r,t)&=&e^{\xi({\bf s},t)-\xi({\bf s},r)}\left(r\,G\left(- t\frac{d}{d t} \right)  -t  G\left(r\frac{d}{d r }\right)\right) \frac{e^{\xi({\bf s},r)-\xi({\bf s},t)}} {r-t}\\
&=&e^{\xi({\bf s},t)-\xi({\bf s},r)}\left(G\left(- t\frac{d}{d t} \right)  -G\left(r\frac{d}{d r }\right)\frac{t}{r}\right) {e^{\xi({\bf s},r)-\xi({\bf s},t)}}\sum_{m=0}^\infty \left(\frac{t}{r}\right)^m\\
&=&\sum_{m=0}^\infty\left(\frac{t}{r}\right)^m   \left(e^{\xi({\bf s},t)}G\left(- t\frac{d}{d t} -m \right) e^{-\xi({\bf s},t)} -  \frac{t}{r}e^{-\xi({\bf s},r)}  G\left(r\frac{d}{d r }-m-1\right) 
e^{\xi({\bf s},r)}\right). \cr
& &
\eea
Then, from (\ref{Pseries}) and (\ref{PPtilde}) we conclude
\bea
A(r,t)= \sum_{m=0}^\infty \left(\frac{t}{r}\right)^m \sum_{N=0}^\infty \sum_{n=0}^N\left( t^N G(-n-m)- r^{N-1}t G(N-n-m-1)\right)  c(N,n)\\
=\sum_{N=0}^\infty \sum_{n=0}^N c(N,n)  \sum_{m=0}^\infty \left(t^{m+N}r^{-m}G(-n-m) -t^{m+1}r^{N-m-1}G(N-n-m-1)\right),
\eea
where we introduced 
\be
c(N,n):=h_n(-{\bf s}) h_{N-n}({\bf s}).
\ee
In the first term we can change the summation variable $m\mapsto m+1-N$ to get
\bea
A(r,t)&=&1+\sum_{N=1}^\infty \sum_{n=0}^N c(N,n)  \sum_{m=N-1}^\infty t^{m+1}r^{N-m-1}G(N-n-m-1)\\
&-&\sum_{N=1}^\infty \sum_{n=0}^N c(N,n)  \sum_{m=0}^\infty t^{m+1}r^{N-m-1}G(N-n-m-1)\\
&=&1-\sum_{N=2}^\infty \sum_{n=0}^{N} c(N,n)  \sum_{m=0}^{N-2} t^{m+1}r^{N-m-1}G(N-n-m-1),
\eea
so that indeed for $i,j>0$
\be
A_{ij}=- \sum_{n=0}^{i+j} G(j-n)h_n(-{\bf s}) h_{i+j-n}({\bf s}).
\label{AijGhh}
\ee
If $M$ and $L$ are finite, then  (\ref{Aseriesdef}) is a polynomial of total degree $LM$, and (\ref{Azero}) follows. 
\end{proof}

For arbitrary $\beta$ we have
\be
A_{ij}= -\sum_{n=0}^{i+j} G(\beta(j-n))h_n(-\beta^{-1}{\bf s}) h_{i+j-n}(\beta^{-1}{\bf s}).
\ee
which coincides with eq~(\ref{Acoefficients}).

In terms of $\tilde{A}$ eq~(\ref{Azero}) is equivalent to
\be
\tilde{A}^{(G, \beta, {\bf s})}_{ij} =0 \quad \text {if} \quad i \ge LM \quad \text{or} \quad j \ge LM.
\ee
and therefore
\be
\tilde{K}^{(G,\beta,\gamma, {\bf s})}_2(z, w) = { \sum _{i=0}^{LM-1}\sum_{j=0}^{LM-1}\tilde{A}^{(G, \beta, {\bf s})}_{ij} w_{-j}(z) w^{*}_{-i}(w)
 -  w_1(z) w_1^{*}(w)  \over z-w}.
\label{pair_correl_expansion_w_rho_s_k_finite} 
\ee
Equivalently, we have the finite form of eq.~(\ref{K_Psi+_Psi-denom})
\be
K(x, x') = {\sum_{i=0}^{LM} \sum_{j=0}^{LM} A_{ij} \Psi^+_i(x) \Psi^-_j(x') \over x - x'}.
\label{K_Psi+_Psi-_finite}
\ee

\subsection{The multicurrent correlators as generating functions for  weighted Hurwitz numbers }
\label{current_correl_hurwitz}

In this section, we specialize the multicurrent correlators  (\ref{multicurrent_correl})
 to the case of  of hypergeometric $\tau$-functions, for which
\be
\hat{g} = \hat{C}_\rho  \hat{\gamma}_-(\beta^{-1} {\bf s}).
\ee
In this case
\be
 \JJ_n(x_1, \dots, x_n) := \langle 0 | \left(\prod_{i=1}^n J_+(x_i)\right) \hat{C}_\rho   \hat{\gamma}_-(\beta^{-1} {\bf s})| 0 \rangle.
 \label{multicurrent_correl_hypergeom}
 \ee 
The correlators $W_n$ are therefore expressed fermionically as
\be
W_n(x_1, \dots, x_n) = {1\over \prod_{i=1}^n x_i}\langle 0 | \prod_{i=1}^n J_+(x_i) \hat{C}_\rho \hat{\gamma}_-(\beta^{-1} {\bf s}) | 0 \rangle.
\label{W_n_fermionic_hypergeometric}
\ee
We now introduce another set of generating functions for weighted Hurwitz numbers, 
both for  the connected and nonconnected case.
\begin{definition}
\bea
F_n({\bf s}; x_1,\dots,x_n) &\&:=  \sum_{\mu,\nu,\,\ell(\mu)=n}  \sum_d \gamma^{|\mu|} \beta^{d-\ell(\nu)} H_{G}^d(\mu,\nu)\, 
|\aut(\mu)| m_\mu(x_1,\dots,x_n) p_\nu({\bf s}) 
\label{F_n_expan}\\
\tilde{F}_n({\bf s};x_1,\dots,x_n) &\& :=  \sum_{\mu,\nu,\,\ell(\mu)=n}  \sum_d \gamma^{|\mu|} \beta^{d-\ell(\nu)} \tilde H^d_{G}(\mu,\nu)\,
|\aut(\mu)|  m_\mu(x_1,\dots,x_n) p_\nu({\bf s}) 
\label{tilde_F_n_expan}\\
&\& = {\hskip 15 pt}\sum_g \beta^{n+2g -2} \tilde{F}_{g, n}, \\
\tilde{F}_{g,n}({\bf s}; x_1, \dots , x_n)  &\&:=  \sum_{\mu,\nu,\,\ell(\mu)=n}   \gamma^{|\mu|} \tilde H^{n + \ell(\nu) + 2g -2}_{G}(\mu,\nu)\, 
|\aut(\mu)| m_\mu(x_1,\dots,x_n) p_\nu({\bf s})\,, \cr
&\&
\label{tilde_F_g_n_expan}
\eea
where $m_\mu(x_1, \dots, x_n)$ denotes the monomial sum symmetric function corresponding to the partition $\mu$,
and $(g, d)$ are related by the Riemann-Hurwitz formula
\be
2-2g = \ell(\mu) + \ell(\nu) - d.
\label{Riem_Hurw_g_d}
\ee
\end{definition}
Note that $F_n$, $\tilde F_n$ belong to $\mathbb{K}[x_1,\dots,x_n; {\bf s}; \beta, \beta^{-1}][[\gamma]]$ 
and $\tilde F_{g,n}$ to $\mathbb{K}[x_1,\dots,x_n; {\bf s}][[\gamma]]$.

We may similarly define the multicurrent correlator $\tilde W_{g,n}({\bf s}; x_1,\dots,x_n)$ for fixed genus $g$ and $n$
as the coefficients in the expansion
\be
\tilde{W}_n({\bf s};x_1,\dots,x_n)=:  \sum_g \beta^{n+2g -2} \tilde{W}_{g, n}({\bf s};x_1,\dots,x_n). 
\ee
The following result, relating  the generating functions $F_n, \tilde{F}_n, \tilde{F}_{g,n}$ to
the multicurrent correlators, is proved in \cite{ACEH1}.
\begin{proposition}
\label{multicurrent_Wn_Fn}
\bea
\label{firW}
W_{n}({\bf s}; x_1,\dots,x_n ) &\&= {\partial^n\over \partial x_1 \cdots \partial x_n}F_{n}({\bf s}; x_1,\dots,x_n), 
\label{W_n_F_n}\\
\tilde W_{n}({\bf s}; x_1,\dots,x_n) &\&={\partial^n\over \partial x_1 \cdots \partial x_n}\tilde F_{n}({\bf s}; x_1,\dots,x_n), 
\label{tilde_W_n_F_n}\\
\tilde W_{g,n}({\bf s}; x_1,\dots,x_n) &\&= {\partial^n\over \partial x_1 \cdots \partial x_n}\tilde F_{g,n}({\bf s}; x_1,\dots, x_n). 
\label{tilde_W_g_n_F_n}\
\eea
\end{proposition}

\section{Bose-Fermi correspondence and  ``cut and join'' operators}
\label{caj}

The cut-and-join description of the single and double Hurwitz numbers was introduced in \cite{GJ1,V}. It is known that generating functions of general weighted Hurwitz numbers (or, equivalently, hypergeometric $\tau$-functions) can be described by several versions of cut-and-join equations, see e.g. \cite{AMMN,ALS,KZ,MMN,Ta1,Sh,AMMN1}. Here we describe two different  cut-and-join representations using the bosonization procedure to express these relations both
bosonically and in terms of fermionic VEV's.


\subsection{Bose-Fermi correspondence and cut and join representation of the  $\tau$-function}
\label{bosonization}

One of the cut-and-join type representations is a direct generalization of \cite{GJ1,V} and is naturally parametrized by the coefficients of $\log(G)$
\be
\sum_{k=0}^{\infty}(-1)^{k+1} \beta^{k} A_k x^k:= \log(\gamma G(\beta x)).
\ee
introduced in Section \ref{hypergeometric_specialization}. From (\ref{Tfuncd}) it follows that element (\ref{A_hat_def}) can be represented as 
\bea
\hat{A}&=&\log\gamma \sum_{j=-\infty}^\infty j :\psi_j \psistar_j:+\sum_{k=1}^\infty (-1)^{k+1} \beta^k A_k \mbox{res}_{z=0}\left( \normord\psi(z) p_k\left(-{\mathcal D}-1\right) \psistar(z)\normord\right)\\
&=&\sum_{k=0}^\infty (-1)^{k+1} \beta^k A_k \hat{Q}_k,
\label{A_k_Q_k_def}
\eea
where the operators
\bea
\hat{Q}_k &\&:=\mbox{res}_{x=o}\left( \normord\psi(x) p_k\left(-x\frac{\pp}{\pp x}-1\right) \psistar(x)\normord\right) 
\label{hatqdef} \\
&\& = \sum_{i =-\infty}^\infty p_k(i) \no{\psi_i \psi^*_i}
\label{hatqdef_expl} 
\eea
commute with each other
\be
\left[\hat{Q}_k,\hat{Q}_m\right]=0, \,\,\,\,\,\,\,\,\,k,m\geq 0.
\ee
From (\ref{A_hat_def}) and (\ref{A_k_Q_k_def}), we may express the exponents $T_i$ in terms of the $A_i$'s as
\be
T_i =\sum_{k=0}^\infty (-1)^{k+1} \beta^k A_k p_k(i), \quad i \in \Zb.
\ee
\

The bose-fermi correspondence is defined by the bosonization map 
\bea
\MM_0: \FF_0 &\& \ra \BB_0 \cr 
\MM_0 |v \rangle &\& \mapsto  \langle 0 | \hat{\gamma}_+({\bf t}) | v \rangle
\eea
from the (zero charge sector of the) Fermi Fock space $\FF_0 \ss \FF$ to the corresponding sector $\BB_0$ of the Bosonic Fock space $\BB_0$.
The latter is viewed as consisting of symmetric functions of an infinite number of auxiliary Bosonic variables
$\{\xi_a\}_{a\in \Nb^+}$, in terms of which the KP flow parameters ${\bf t } =(t_1, t_2, \dots )$ are understood as normalized power sums
\be
t_i :={1\over i} \sum_{a=1}^\infty  \xi_a^i, \quad i =1 , \dots. 
\ee

This  gives rise to the identity
\be\label{bfcorresp}
 \langle 0 | \hat{\gamma}_+({\bf t}) \normord \psi(z) \psistar(w) \normord | v\rangle =\frac{1}{z-w}  \left(\normordboson e^{\phi(z)-\phi(w)} \normordboson-1\right)\, \langle 0 | \hat{\gamma}_+({\bf t}) |v\rangle, \quad \forall \ |v\rangle \in \FF_0,
\ee
where $\normordboson \dots\normordboson$ denotes the bosonic normal ordering, which puts all derivatives with respect to $t_k$ to the right of all multiplications by $t_k$'s, and ${\phi}$ is the bosonic operator
\be
{\phi}(x):=\sum_{k=1}^\infty \left( t_k x^k -\frac{1}{k x^k} \frac{\pp}{\pp t_k} \right).
\label{phi_x}
\ee

This allows us to construct a bosonic counterpart for any bilinear fermionic combination $\sum a_{ij}  \normord \psi_i \psistar_j  \normord$. 
In particular, for the operators (\ref{hatqdef}), we have the bosonic representation $Q_k$, which are polynomials in $t_k$ and $\frac{\pp}{\pp t_k}$, such that
\be
 \langle 0 | \hat{\gamma}_+({\bf t})\hat{Q}_k=Q_k\,  \langle 0 | \hat{\gamma}_+({\bf t}).
 \label{qhattoq}
\ee
From (\ref{bfcorresp}) it immediately follows \cite{Sh,ARP} that a generating function of the bosonic counterparts of these operators $Q_k$ is
\be
\sum_{k=0}^\infty \frac{a^k}{k!}{Q}_k=\frac{1}{(e^{a/2}-e^{-a/2})^2}\,\res_{x=0}\left(x^{-1}\left(\normordboson e^{{\phi}(x e^{a/2})-{\phi}(x e^{-a/2})}\normordboson-1\right)\right).
\label{Q_k_bosonic_def}
\ee

The $Q_k$'s also generate a commutative algebra
\be
\label{commutbosq}
\left[{Q}_k,{Q}_m\right]=0, \,\,\,\,\,\,\,\,\,k,m\geq 0.
\ee 
For example,
\bea
{Q}_0&=&\sum_{k=1}^\infty k t_k \frac{\pp}{\pp t_k},\\
{Q}_1&=& \frac{1}{2}\sum_{a,b}
\left(ab t_at_b\frac{\pp}{\pp t_{a+b}} + (a+b) t_{a+b}\frac{\pp^2}{\pp t_a\pp t_b}\right),  \\
{Q}_2&=& \frac{1}{3} \sum_{a,b,c=1}^\infty \left(abc\,t_at_bt_c \,\frac{\pp}{\pp t_{a+b+c}}
+ (a+b+c)\, t_{a+b+c}\,
\frac{\pp^3}{\pp t_a\pp t_b\pp t_c} \right)\cr
&+& \ \frac{1}{2} \sum_{a,b=1}^\infty \,\sum_{c=1}^{a+b-1} c (a+b-c)\, t_ct_{a+b-c}\,\frac{\pp^2}{\pp t_a\pp t_b}
+\ \frac{1}{6}\sum_{a=1}^\infty a(a^2-1)\,t_{a}\,\frac{\pp}{\pp t_a}.  
\eea

From (\ref{qhattoq}) we have
\be
\langle 0 | \hat{\gamma}_+({\bf t})e^{\hat{A}}=\exp\left(\sum_{k=0}^\infty (-1)^{k+1} \beta^k A_k {Q}_k \right) \langle 0 | \hat{\gamma}_+({\bf t}).
\ee
This gives a cut-and-join type representation for the $\tau$-function (\ref{tau_G_H}):
\bea
\tau^{( G, \beta,\gamma)}({\bf t}, {\bf s})&=&  \langle 0 | \hat{\gamma}_+({\bf t})e^{\hat{A}} \hat{\gamma}_-({\bf s})|0\rangle\\
&=&\exp\left(\sum_{k=0}^\infty (-1)^{k+1} \beta^k A_k {Q}_k \right) \langle 0 | \hat{\gamma}_+({\bf t}) \hat{\gamma}_-({\bf s})|0\rangle\\
&=&\exp\left(\sum_{k=0}^\infty (-1)^{k+1} \beta^k A_k {Q}_k \right) \exp{\left(\sum_{k=1}^\infty k t_k s_k\right)}.
\label{tau_expAK}
\eea
For the case of simple Hurwitz numbers, considered in Section \ref{simple_hurwitz} when $G(z)=\exp(z)$, if we put $\gamma=1$, (\ref{tau_expAK}) reduces to 
\be
\tau({\bf t}, {\bf s})=e^{\beta Q_2} e^{\left(\sum_{k=1}^\infty k t_k s_k\right)},
\ee
which coincides with the cut-and-join representation of \cite{GJ1,V} for simple (double) Hurwitz numbers.


\subsection{Cut and join equations}
\label{cut_and_jon_eq}

If the $A_k$'s are viewed as the independent parameters defining the weighting, 
the $\tau$-function $\tau^{(G, \beta, \gamma)}({\bf t}, {\bf s})$ satisfies the generalized cut-and-join equation
\bea
\beta \frac{\pp}{\pp \beta} \tau^{( G, \beta,\gamma)}({\bf t}, {\bf s}) =\sum_{k=1}^\infty k A_k \frac{\pp}{\pp A_k}\, \tau^{( G, \beta,\gamma)}({\bf t}, {\bf s}).
\label{cut_and_join}
\eea
It follows from (\ref{qhattoq}) and the commutation relations (\ref{commutbosq}) that $\tau^{( G, \beta,\gamma)}({\bf t}, {\bf s})$ also  satisfies an 
infinite number of further cut-and-join equations corresponding to the variables $A_k$ and $\gamma$
\bea
-\gamma \frac{\pp}{\pp \gamma}\tau^{( G, \beta,\gamma)}({\bf t}, {\bf s}) &\& =-{Q}_0\, \tau^{( G, \beta,\gamma)}({\bf t}, {\bf s}),
\label{gamma_derive}\\
\frac{\pp}{\pp A_k}\tau^{( G, \beta,\gamma)}({\bf t}, {\bf s}) &\&= (-1)^{k+1} \beta^k {Q}_k \tau^{( G, \beta,\gamma)}({\bf t}, {\bf s}).
\label{tau_beta_deriv}
\eea

For the case where $G(z)$ is a polynomial, we have an alternative cut-and-join representation for the single Hurwitz numbers \cite{Z,KZ}.  Namely, 
\be
e^{\hat{A}} e^{J_{-1}} | 0\rangle=e^{ e^{\hat{A}}J_{-1} e^{-\hat{A}}} | 0\rangle.
\ee
Here
\bea
 e^{\hat{A}}J_{-1} e^{-\hat{A}}&=& e^{\hat{A}} \res_{z=0} \left(z^{-1} \psi(z) \psistar(z)\right)e^{-\hat{A}}\\
& =& \res_{z=0} \left(z^{-1} (e^{T\left({\mathcal D}\right)} \psi(z)) e^{-T\left(-{\mathcal D}-1\right)} \psistar(z))\right)\\
& =& \res_{z=0} \left( (e^{-T\left({\mathcal D}\right)}z^{-1} e^{T\left({\mathcal D}\right)} \psi(z))  \psistar(z)\right)\\
& =& \res_{z=0} \left( (R^* \psi(z))  \psistar(z)\right)\\
& =& \gamma J_{-1} +\gamma  \sum_{k=1}^{M} g_k \beta^k  \res_{z=0} \left(z^{-1} ({\mathcal D} ^k \psi(z))  \psistar(z)\right),
\eea
where we have used (\ref{com}). Thus
\be
\tau^{( G, \beta,\gamma)}({\bf t}, {\bf s=\delta_{k,1}})=\exp\left(\gamma\left( t_1+ \sum_{k=1}^{M} g_k \beta^k V_k \right)\right)\cdot 1
\ee
where $V_k$ are the bosonic operators corresponding to $ \res_{z=0} \left( z^{-1}({\mathcal D} ^k \psi(z))  \psistar(z)\right)$. For example
\bea
V_1&=&\sum_{k=1}^\infty k t_k \frac{\pp}{\pp t_{k-1}},\\
V_2&=&\sum_{k,l=1}^\infty \left(kt_k lt_l \frac{\pp}{\pp t_{k+l-1}}+(k+l+1)t_{k+l+1}\frac{\pp^2}{\pp t_k \pp t_l}\right).
\eea

\section{Examples}
\label{examples}

To conclude, we focus  upon  four specific examples of weighted Hurwtz numbers: the simple (single and double) Hurwitz numbers of \cite{Pa,Ok}}; the strongly monotonic Hurwitz numbers,  corresponding to weights supported curves with theree branch points (Belyi curves)  \cite{AC1, KZ, H2}; the weakly monotonic Hurwitz numbers, corresponding to signed weightings  \cite{GGN1, GH2, H2}} and the quantum Hurwitz numbers introduced in  \cite{GH2, H2}}.  Explicit formulae are provided  for: the content product coefficients  entering in the $\tau$-function series,  the parameters defining their fermionic respresentation, the weighting factors  defining the Hurwitz numbers, and the classical and quantum spectral curves  entering in the topological recursion approach.

\subsection{Simple (single and double) Hurwitz numbers \cite{Pa,Ok}}
\label{simple_hurwitz}

This is the original case studied  by Pandharipande and Okounkov. It corresponds to the exponential
weight generating function
\be
G(z) = Exp(z) = e^z.
\ee
The content product formula entering in the Schur function series (\ref{tau_G_H}), (\ref{tau_tilde_G_H}) for the $\tau$-function is
\be
r_\lambda^{(Exp, \beta)}= e^{{\beta\over 2} \sum_{i=1}^{\ell(\lambda)} \lambda_i(\lambda_i - 2i +1)}.
\ee
The weight $W_{Exp}(\mu^{(1)}, \dots, \mu^{(k)})$ is the Dirac measure (i.e. the characteristic function) supported on the
partitions consisting only of $2$-cycles
\be
\mu^{(i)} = (2, (1)^{N-2}),\quad  i=1, \dots , k,
\ee
\be 
W_{Exp}(\mu^{(1)}, \dots, \mu^{(k)}) = \prod_{i=1}^k \delta_{(\mu^{(i)}, (s, (1)^{N-2})}
\ee
and the exponents $T_i(\beta, \gamma)$ entering in (\ref{rho_i_T_i}), (\ref{rho_j_gamma_G}) and (\ref{g_C_hat_rho_gamma-}) are
\be
T_i(\beta, \gamma) = i \ln(\gamma) + {{\beta\over 2} i(i+1)}.
\ee
The quantum spectral curve is \cite{ALS}:
\be
\beta D  \Psi^+_0(x)  - \sum_{i} i s_i \gamma^i x^i e^{i\beta D} \Psi^+_0(x)=0
\ee
where $D=x{d \over dx}$ is the Euler operator, and
the classical spectral curve is
\be
y = \sum_{i} i s_i \gamma^i x^{i-1} e^{ixy}. 
\ee

\subsection{Three branch points (Belyi curves): strongly monotonic paths in the Cayley graph \cite{AC1, KZ, H2}}
\label{belyi_curves}
For this case, the weight generating function is linear
\be
G(z) = 1+z =:E(z).
\ee
The content product formula is
\be
r_\lambda^{(E, \beta)} = \beta^{|\lambda|} \left( 1/ \beta \right)_{\lambda},
\ee
where
\bea
u_\lambda &\&:= \prod_{i=1}^{\ell(\lambda)} (u -i +1)_{\lambda_i} \cr
(u)_j &\& := u (u+1) \cdots (u+j-1)
\eea
is the multiple Pochhammer symbol corresponding to the partition $\lambda$.
The measure $W_E(\mu^{(1)}, \dots, \mu^{(k)})$ is supported on the set of $k=1$ partitions,
and hence corresponds to enumeration of coverings with just three branch points $(\mu^{(1)}, \mu, \nu)$ (Belyi curves).
\be
W_E(\mu^{(1)}, \dots, \mu^{(k)}) = {\delta_{k, 1} \over \ell^*(\mu^{(1)})}.
\ee
The exponents $T_i(\beta)$ entering in (\ref{rho_i_T_i}), (\ref{rho_j_gamma_G}) and (\ref{g_C_hat_rho_gamma-}) are
\be
T_i^{E}(\beta, \gamma) =  i \ln(\gamma)+\sum_{j=1}^i \ln(1+j \beta), \quad T_{-i}^{E}(\beta) = -i \ln(\gamma) -\sum_{j=1}^{i-1} \ln(1-j \beta), \quad i>0
\ee

The quantum spectral curve is \cite{ALS}:
\be
\beta D  \Psi^+_0(x)  - \sum_i  i s_i \gamma^i x^i (1 + \beta D)^i \Psi^+_0(x)=0
\ee
and the classical one is
\be
xy = \sum_i i s_i \gamma^i \left( x (1+xy)\right)^i.
\ee
\subsection{Signed Hurwitz numbers: weakly monotonic paths in the Cayley graph \cite{GGN1, GH2, H2}}
\label{signed_hurwitz}
In this case the weight generating function is
\be
\tilde{G}(z) = {1\over 1-z} =: H(z).
\ee
The content product coefficient is
\be
r_\lambda^{(H, \beta)} =(- \beta)^{-|\lambda|} \left( (-1/\beta)\right)_{\lambda}^{-1}.
\ee
The measure $W_H(\mu^{(1)}, \dots, \mu^{(k)}) $ is thus given by the evaluation of the ``forgotten'' symmetric functions
$f_\lambda({\bf c})$ for the partition with parts $\{\lambda_i\}= \{\ell^*(\mu^{(i)})\}$ at ${\bf c} = (1, 0, \dots )$
\be
W_H(\mu^{(1)}, \dots, \mu^{(k)}) = (-1)^{d+k} {k! \over \prod_{i=1}^k m_i(\lambda)!}
\ee
where $m_i(\lambda)$ is the number of colengths $\ell^*(\mu^{(j)})$ equal to $i$.
The exponents $T_i(\beta)$ entering in (\ref{rho_i_T_i}), (\ref{rho_j_gamma_G}) and (\ref{g_C_hat_rho_gamma-}) are
\be
T_i^{H}(\beta, \gamma) :=  i \ln(\gamma)-\sum_{j=1}^i \ln(1-j \beta), 
\quad T_{-i}^{H}(\beta) = -i \ln(\gamma) +\sum_{j=1}^{i-1} \ln(1+j \beta), \quad i>0
\ee
The quantum spectral curve is \cite{ALS}:
\be
\beta D  \Psi^+_0(x)  - \sum_i  i s_i \gamma^i x^i (1 - \beta D)^{-i} \Psi^+_0(x)=0
\ee
and  the classical one is
\be
xy = \sum_i i s_i \gamma^i x^i (1-xy)^{-i}.
\ee

\subsection{Simple quantum Hurwitz numbers \cite{GH2, H2}}
\label{quantum_hurwitz}
Choosing a real parameter $q$, $0\le q \le 1$, we select the parameters $c_i$ to be
\be
c_i = q^i, \quad i =1, 2, \dots.
\ee
The weight generating function for this case is therefore
\be
G(z)= \prod_{i=1}^\infty (1 + q^i z)=: E'(q) = {1\over 1+z} e^{-\left( {1\over 1-q}\right)Li_2(q, -z)},
\ee
where $\Li_2(q, z)$ is the quantum dilogarithm function
\be
\Li_2(q, z) =(1-q) \sum_{i=1}^\infty {z^i \over i (1- q^i)}.
\ee
This gives the ``simple quantum'' Hurwitz numbers introduced in \cite{GH2}, \cite{H2}. 
The content product  coefficient for this case is
\bea
r^{(E'(q), \beta)}_j &\&= \prod_{k=1}^\infty (1+ q^k \beta j) = (-q\beta j; q)_{\infty} , \\
r^{(E'(q), \beta)}_\lambda(z) &\&= \prod_{k=1}^\infty \prod_{(i,j)\in \lambda} (1+ q^k \beta (j-i)) 
 = \prod_{(i,j)\in \lambda} (-q\beta(j-i); q)_\infty \cr
&\& = \prod_{k=1}^\infty (\beta q^k)^{\abs{\lambda}} (1/(\beta q^k))_\lambda, 
\eea
The weight  $W_{E'(q)}(\mu^{(1)}, \dots, \mu^{(k)})$ for this case is
\bea
W_{E'(q)} (\mu^{(1)}, \dots, \mu^{(k)}) &\& := m_\lambda (q, q^q, \dots )\cr
&\& =  {1\over  |\aut(\lambda)|} \sum_{\sigma\in S_k} \frac{q^{k \ell^*(\mu^{(\sigma(1))})} \cdots q^{\ell^*(\mu^{(\sigma(k))})}}{
(1- q^{\ell^*(\mu^{(\sigma(1))})}) \cdots (1- q^{\ell^*(\mu^{(\sigma(1))}} \cdots q^{\ell^*(\mu^{(\sigma(k))})})} \cr
&\& =  {1\over  |\aut(\lambda)|} \sum_{\sigma\in S_k} \frac{1}{
(q^{-\ell^*(\mu^{(\sigma(1))})} -1) \cdots (q^{-\ell^*(\mu^{(\sigma(1))})} \cdots q^{-\ell^*(\mu^{(\sigma(k))})}-1)}, \cr
&\&
\label{W_Eprime_q}
\eea
which may be viewed as the (unnormalized) probability measure of a Bosonic gas with energy levels
\be
\EE(\mu) := \ell^*(\mu)\hbar \omega_0
\ee
associated to branch points with ramification profiles $\mu$, where
\be
q = e^{-\beta \hbar \omega_0}.
\ee

The exponents $T_i(\beta)$ entering in (\ref{rho_i_T_i}), (\ref{rho_j_gamma_G}) and (\ref{g_C_hat_rho_gamma-}) are
\bea
T_i^{E'(q)}(\beta, \gamma) &\&:=  i \ln(\gamma)+  \sum_{k=1}^i \sum_{j=1}^\infty \ln(1+ k\beta q^j)  \\
 T_{-i}^{E'(q)}(\beta, \gamma) &\& = -i \ln(\gamma) -  \sum_{k=1}^{i-1} \sum_{j=1}^\infty \ln(1- k\beta q^j) \quad i>0
\eea

The quantum spectral curve is:
\be
\beta D  \Psi^+_0(x)  - \sum_i  i s_i \gamma^i x^i \left(\prod_{j=1}^\infty1 -q^j \beta D\right)^i \Psi^+_0(x)=0
\ee
and  the classical one is
\be
xy = \sum_i i s_i \gamma^i x^i  (\prod_{j=1}^\infty1 -q^j xy)^i .
\ee

\bigskip


\begin{appendices}


\section{Fermionic representation of multipair correlators}
 \label{AppendA}
 
 In the following, we  assume that $g$ is lower triangular. This implies, in
particular, that the left vacuum is stabilized by both $\hat{g}$ and $\hat{g}^{-1}$
\be
\langle 0 | \hat{g} = \langle 0 | \hat{g}^{-1} = \langle 0 |.
\ee

\subsection{Fermionic representation of  $\tilde{K}_2^g(z, w)$ }
\label{fermionic_VEV_K2}

We provide here and in Appendix  \ref{K_g_2n_tau} the details of the proofs of Proposition \ref{K_2_tau_prop},
Lemma \ref{main_fermionic_identity} and Corollaries \ref{corollary_2n}, \ref{K2n_det}.
As defined in eq.~(\ref{K2_VEV_rep}), the pair correlator $\tilde{K}^g_2(z,w)$ 
may be expressed fermionically as
\be
\tilde{K}^g_2(z,w) =
\begin{cases}  \langle 0 | \hat{g}^{-1} \psi(z) \psi^*(w) \hat{g} | 0 \rangle \quad \text{if} \quad |z|>|w| \cr
 - \langle 0 | \hat{g}^{-1} \psi^*(w) \psi(z) \hat{g} | 0 \rangle \quad \text{if} \quad |z|<|w|.
 \end{cases}
\label{K_2_fermionic}
\ee
We now prove Proposition \ref{K_2_tau_prop}, which states that this is equivalent to
\be
\tilde{K}^g_2(z,w) = { \tau_g( [w^{-1}]-[z^{-1}]) \over z-w}. 
\label{K_2_tau_appB}
\ee
 The general case will then follow from Wick's theorem, as shown in the next subsection.
 \begin{proof}
 We begin by proving the stronger relations,  
for all partitions $\lambda$,
 \bea
{s_\lambda([w^{-1}] - [z^{-1}]) \over z-w} &\&= {\langle 0 | \hat{\gamma}_+([w^{-1}]-[z^{-1}])|\lambda \rangle\over z-w}  =
 \begin{cases}
\langle 0 | \psi(z) \psi^*(w) | \lambda \rangle, \quad \text{if} \quad |z| > |w|,
 \cr
-\langle 0 | \psi^*(w) \psi(z) | \lambda \rangle, \quad \text{if} \quad |z| < |w|,   
  \label{s_jwz}
  \end{cases} 
 \eea
 which, since these are valid for all $\lambda$, implies
  \be
\frac{ \langle 0 | \hat{\gamma}_+([w^{-1}] -[z^{-1}]) }{z-w} =
 \begin{cases}
\langle 0 | \psi(z) \psi^*(w)  , \quad \text{if} \quad |z| > |w|,
 \cr
-\langle 0 | \psi^*(w) \psi(z) , \quad \text{if} \quad |z| < |w|,  \end{cases}
 \label{gamma_psi_psi_wz}
 \ee
First, note that  the LHS of (\ref{s_jwz}) vanishes unless $\lambda$ is either the trivial partitition $\lambda = \emptyset$ ,
in which case
\be
s_\emptyset([w^{-1}] - [z^{-1}]) = 1,
\ee
or a hook partition $\lambda = (a | b)$, with arm length $a\ge 0$ and leg length $b \ge 0$, for which
\be
s_{(a|b)}([w^{-1}] - [z^{-1}]) = (-1)^b (z-w) z^{-b-1} w^{-a-1}.
\ee

Turning to the RHS, by Wick's theorem, $\langle 0 | \psi(z) \psi^*(w)| \lambda \rangle$  and 
$\langle 0 |  \psi^*(w) \psi(z)| \lambda \rangle$ also vanish unless $\lambda = \emptyset$ or $\lambda = (a| b)$. 
For the first case, if $|z|>|w|$
\bea\label{psipsist}
\langle 0 | \psi(z) \psi^*(w) | 0 \rangle &\& = \sum_{i\in \Zb} \sum_{j\in \Zb} z^i w^{-j-1} \langle 0 | \psi_i \psi^*_j | 0 \rangle \cr
&\& = {1\over w} \sum_{i=1}^\infty \left({w\over z}\right)^i = {1 \over z-w},
\eea
 and similarly, if $|z|<|w|$
 \bea
 \langle 0 | \psi^*(w)\psi(z)   | 0 \rangle &\& = \sum_{i\in \Zb} \sum_{j\in \Zb} z^i w^{-j-1} \langle 0 |\psi^*_j   \psi_i  | 0 \rangle \cr
&\& = {1\over w} \sum_{i=0}^\infty \left({z\over w}\right)^i = {1 \over w-z}.
\eea
For $\lambda = (a|b)$, and $|z|> |w|$
\bea
\langle 0 | \psi(z) \psi^*(w) | (a|b)\rangle &\& 
= \sum_{i \in \Zb} \sum_{j\in \Zb} z^i w^{-j-1} \langle 0 | \psi_i \psi^*_{j}(-1)^b \psi_a \psi^*_{-b-1} | 0\rangle \cr
&\& = (-1)^b w^{-a-1} z^{-b-1}
\eea
by Wick's theorem, while for $|z|<|w|$,
\bea
\langle 0 | \psi^*(w) \psi(z)  | (a|b)\rangle &\& 
= \sum_{i \in \Zb} \sum_{j\in \Zb} z^i w^{-j-1} \langle 0 | \psi^*_{j} \psi_i (-1)^b \psi_a \psi^*_{-b-1} | 0\rangle \cr
&\& = - (-1)^b w^{-a-1} z^{-b-1}. 
\eea
Taking the scalar product of (\ref{gamma_psi_psi_wz}) with $\hat{g}|0 \rangle$ then proves the
equivalence of eq.~(\ref{K_2_fermionic}) with (\ref{K_2_tau_appB}).

 \end{proof}

     
 \subsection{The $n$-pair correlator $\tilde{K}^{g}_{2n}({\bf z}, {\bf w})$}
 \label{K_g_2n_tau}
  For a set of  $n$ nonzero complex numbers ${\bf z} = (z_1, z_2, \dots z_n )$,  we use the
  notations $[{\bf z}]$ and $[{\bf z}^{-1}]$ to denote the infinite sequence $\{[{\bf z}]_i\}_{i=1, 2, \cdots}$, 
  of normalized power sums
  \be
  [{\bf z}]_i :=  {1\over i} \sum_{j=1}^n(z_j)^i  \quad \text{and} \quad    [{\bf z}^{-1}]_i :=  {1\over i} \sum_{j=1}^n(z^{-1}_j)^i
  \ee
    We now proceed to the proof of Lemma \ref{main_fermionic_identity}, which is:
    \begin{lemma}
  \label{main_fermionic_identity_app}
  For $2n$ distinct nonzero complex parameters $\{z_i, w_i\}_{i=1, \dots n}$, we have
  \be
  \langle 0 | \prod_{i=1}^n \left(\psi(z_i) \psi^*(w_i)\right) =
  \det \left({1\over z_i - w_j} \right)_{1\le i,j \le n}\langle 0 | \hat{\gamma}_+([{\bf w}^{-1}] - [{\bf z}^{-1}]),
  \label{psi_psi_dag_gamma_+A}
  \ee
  \end{lemma}
  and Corollaries \ref{corollary_2n}, \ref{K2n_det}.
  \begin{proof}
  We proceed by induction on $n$, recalling that the Cauchy determinant is given by
  \be
  \det \left({1\over z_i - w_j} \right) = (-1)^\frac{n(n-1)}{2} {\Delta ({\bf z}) \Delta({\bf w}) \over \prod_{1 \le i, j, \le n} (z_i - w_j)}
  \label{Cauchy_det}
  \ee
  where
  \be
  \Delta({\bf z}) = \prod_{1\le i < j \le n} (z_i-z_j), \quad  \Delta({\bf w}) = \prod_{1\le i < j \le n} (w_i-w_j).
  \ee
  
 The case $n=1$ is what underlies Proposition \ref{K_2_tau_prop}, which was proved  in Appendix \ref{fermionic_VEV_K2}.
Now assume it is valid for $n-1$
\be
  \langle 0 | \prod_{i=1}^{n-1} \left(\psi(z_i) \psi^*(w_i)\right) =
  \det \left({1\over z_i - w_j} \right)_{1\le i,j \le n-1}\langle 0 | \hat{\gamma}_+([{\bf w}^{-1}] - [{\bf z}^{-1}]).
  \label{induct_hypoth}
  \ee
  and multiply both sides on the right by $\psi(z_n) \psi^*(w_n)$.
  The RHS becomes
  \be
    \det \left({1\over z_i - w_j} \right)_{1\le i,j \le n-1}\langle 0 | \hat{\gamma}_+([{\bf w}^{-1}] - [{\bf z}^{-1}])\psi(z_n) \psi^*(w_n).
    \label{induct_RHS}
  \ee
  Using (\ref{fermiconj}) and
  \bea
   e^{\sum_{i=1}^\infty {1\over i}({z_n \over w_j}) ^i} &\&= 1 -{z_n\over w_j}, \quad 
   e^{-\sum_{i=1}^\infty {1\over i}({w_n \over w_j} )^i} =  {1 \over 1 -{w_n\over w_j}}  \\
    e^{-\sum_{i=1}^\infty  {1\over i}({z_n \over z_j}) ^i} &\&= { 1 \over 1 -{z_n\over z_j}}, \quad 
   e^{\sum_{i=1}^\infty {1\over i} ({w_n \over z_j} )^i} =  1 -{w_n\over z_j} 
   \eea
   we have
   \be
     \hat{\gamma}_+([{\bf w}^{-1}] - [{\bf z}^{-1}])\psi(z_n) \psi^*(w_n) 
     =  \prod_{j=1}^{n-1} {(z_j - z_n) (w_j - w_n) \over (z_j -w_n) (z_n  - w_j)}
   \,   \psi(z_n) \psi^*(w_n)  \hat{\gamma}_+([{\bf w}^{-1}] - [{\bf z}^{-1}]).
     \ee
Substituting this into the RHS of \eqref{induct_RHS} of the inductive hypothesis
and using \eqref{Cauchy_det} for the Cauchy determinant, we obtain \eqref{psi_psi_dag_gamma_+A}.
  \end{proof}
  \medskip
  \begin{remark}
The usual fermionic representation of the Baker function  $\Psi_g(z, {\bf t})$ and its dual $\Psi_g^*(z, {\bf t})$ \cite{SS, Sa}
can be seen to  follow as a corollary of the above.
\medskip
\begin{corollary}
\be
\Psi_g(z, {\bf t}) = {\langle 0 | \psi^*_{0} \hat{\gamma}_+ ({\bf t}) \psi(z) \hat{g}|0 \rangle \over \tau_g({\bf t})},
\quad \Psi_g^*(z, {\bf t}) = {\langle 0 | \psi_{-1} \hat{\gamma}_+ ({\bf t}) \psi^*(z) \hat{g}|0 \rangle \over \tau_g({\bf t})}.
\label{fermionic_Psi}
\ee
\end{corollary}
\begin{proof}
Taking the  $n=1$ case of \eqref{psi_psi_dag_gamma_+A}, evaluating the residue at $z=0$ and setting $w=z$,
or the residue at $w=0$,  gives the identities 
\be
\langle 0 | \gamma_+(-[z^{-1}]) = \langle 0 | \psi^*_{0} \psi(z), \quad
\langle 0 | \gamma_+([z^{-1}]) = \langle 0 | \psi_{-1} \psi^*(z).
\ee
Substituting these into Sato's formulae for  $\Psi_g(z, {\bf t})$ and  $\Psi_g^*(z, {\bf t})$ gives
\bea
\Psi_g(z, {\bf t}) = e^{\xi( {\bf t}, z)}{\langle 0 | \hat{\gamma}_+(-[z^{-1}])  \hat{\gamma}_+({\bf t}) \hat{g} | 0 \rangle \over \tau_g({\bf t})}=
 e^{\xi({\bf t}, z)} {\langle 0 | \psi^*_{0}  \psi(z) \hat{\gamma}_+({\bf t}) \hat{g} | 0 \rangle \over \tau_g({\bf t})} \cr
 \Psi^*_g(z, {\bf t})
  = e^{-\xi({\bf t}, z)}{\langle 0 | \hat{\gamma}_+([z^{-1}])  \hat{\gamma}_+({\bf t}) \hat{g} | 0 \rangle \over \tau_g({\bf t})}=
 e^{-\xi({\bf t}, z)} {\langle 0 | \psi_{-1}  \psi^*(z) \hat{\gamma}_+({\bf t}) \hat{g} | 0 \rangle \over \tau_g({\bf t})} \
\eea
Substituting the relations (\ref{fermiconj}) gives the result \eqref{fermionic_Psi}.
\end{proof}

  \end{remark}

\subsection{Expansion of $\tilde{K}_2(z,w)$  in the adapted basis}
\label{K2_expansion_adapted_basis}

We now prove Proposition \ref{K_w_j_series}, i.e. the expansions
\be
\tilde{K}^g_2(z, w) = \begin{cases}  \sum_{j=1}^\infty w^g_j(w)   w^{g*}_{-j+1}(z) \quad \text{if} \ |z| > |w|, \\
=- \sum_{j=1}^\infty w^g_{-j+1}(w)   w^{g*}_{j}(z) \quad \text{if} \ |z| < |w|.
\end{cases}
\label{pair_correl_expansion_w_g_k_appB}
\ee

This can be done in two ways: the first is a direct calculation, using the expression (\ref{K_2_tau_appB}) for $K_g(z,w)$ in terms
of the $\tau$-function, to obtain a double power series expansion 
in $z$ and $w$ by specializing the Schur function expansion  of $\tau_g({\bf t})$ at these values and the fact that the
affine coordinates are just the hook Pl\"ucker coordinates. The second is  a fermionic calculation,
based  on again introducing sums over a complete set of intermediate states. We  here present the details
of the latter.
\begin{proof} Proof of (\ref{pair_correl_expansion_w_g_k_appB}).
For $|z|> |w|$, inserting a sum over a complete set of intermediate states, we have
\be
\langle 0 |\hat{g}^{-1} \psi(z) \psi^*(w) \hat{g} |0 \rangle = \sum_{\lambda} \langle 0 | \hat{g}^{-1}\psi(z) \hat{g} | \lambda; -1 \rangle
\langle \lambda; -1 | \hat{g}^{-1} \psi^*(w) \hat{g}|0 \rangle
\ee
By Wick's theorem, the only partitions that contribute to this sum are those of the form $\lambda = (1)^j$, $j=0, 1, \dots$.
Substituting
\be
|(1)^j ; -1\rangle = (-1)^j \psi^*_{-j-1} | 0 \rangle,
\ee
into the sum therefore gives
\bea
\langle 0 |\hat{g}^{-1} \psi(z) \psi^*(w) \hat{g} |0 \rangle &\& = \sum_{j=0}^\infty \langle 0| \hat{g}^{-1} \psi(z) \hat{g} \psi^*_{-j-1}| 0 \rangle 
 \langle 0 | \psi_{-j-1}  \hat{g}^{-1} \psi^*(w) \hat{g} | 0 \rangle \cr
 &\& = \sum_{j=1}^\infty w_j^g(w) w^{g*}_{-j+1}(z). 
\eea
A similar calculation shows, for $|z|< |w|$
\bea
\langle 0 |\hat{g}^{-1} \psi^*(w)\psi(z)  \hat{g} |0 \rangle &\& = \sum_{j=0}^\infty \langle 0| \hat{g}^{-1} \psi^*(w) \hat{g} \psi_j| 0 \rangle 
 \langle 0 | \psi^*_j  \hat{g}^{-1} \psi(z) \hat{g} | 0 \rangle \cr
 &\& = \sum_{j=1}^\infty w_{-j+1}^g(w) w^{g*}_{j}(z). 
\eea
\end{proof}

      \end{appendices}

 \bigskip
\noindent 
\small{ {\it Acknowledgements.} 
The work of A.A. is supported by IBS-R003-D1 and  by RFBR grant 15-01-04217. 
B.E. wishes to thank the Centre de recherches math\'ematiques, Montr\'eal, for the Aisenstadt Chair grant, 
and the FQRNT grant from the Qu\'ebec government, that partially supported this joint project.
G.C. acknowledges support from the European Research Council, grant ERC-2016-STG 716083 ``CombiTop'', from the Agence Nationale de la Recherche, grant ANR 12-JS02-001-01 ``Cartaplus'' and from the City of Paris, grant ``\'Emergences 2013, Combinatoire \`a Paris''.
The work of J.H. was partially supported by the Natural Sciences and Engineering Research Council of Canada (NSERC) and the Fonds de recherche du Qu\'ebec, Nature et technologies (FRQNT).  He wishes to thank the Institut des Hautes \'Etudes Scientifiques, where much of this work was completed. It was also supported by the ERC Starting Grant no. 335739 ``Quantum fields and knot homologies'' funded by the European Research Council under the European Union's  Seventh Framework Programme.  The authors would also like to thank the organizers of the 
January - March, 2017  thematic semester ``Combinatorics and interactions''  at the Institut Henri Poincar\'e, where they were participants during the completion of this work.}


\newcommand{\arxiv}[1]{\href{http://arxiv.org/abs/#1}{arXiv:{#1}}}

\bigskip
\noindent

\end{document}